\documentclass[year=26,pdfa]{fmcad}

\usepackage{physics}
\usepackage{amsmath, amssymb, amsthm}
\usepackage[T1]{fontenc}
\usepackage{graphicx}
\usepackage{color}

\usepackage{doi}
\usepackage{caption}
\usepackage{siunitx}
\usepackage{booktabs}
\usepackage{threeparttable}
\usepackage{multirow}
\usepackage{subcaption}
\usepackage{csquotes}
\MakeOuterQuote{"}
\usepackage{refcount}

\newtheorem{theorem}{Theorem}
\newtheorem{lemma}{Lemma}
\theoremstyle{definition}
\newtheorem{example}{Example}

\newcommand{\set}[2]{\{ #1 \;|\; #2 \}}
\newcommand{\Fam}{\mathcal{F}}
\newcommand{\iso}{\cong}
\newcommand{\eqc}[1]{\left[ #1 \right]}
\newcommand{\ext}{\texttt{ext}}
\newcommand{\comp}{\texttt{comp}}
\newcommand{\prfsys}{\textsc{symbreak}}
\newcommand{\branch}{\textbf{Branch}}
\newcommand{\prune}{\textbf{Prune}}
\newcommand{\node}[2]{\langle #1, #2 \rangle}
\newcommand{\Nodes}{\mathcal N}

\newcommand{\nat}{\mathbb{N}}

\ExplSyntaxOn

\prop_new:N \g_restatable_prop

\NewDocumentEnvironment{restatable}{ m m +b }
 {
  \prop_gput:Nnn \g_restatable_prop { #2 } { { #1 } { #3 } }
  
  \begin{#1} \label{#2} #3 \end{#1}
 }
 { }

\NewDocumentCommand{\restate}{ m }
 {
  \group_begin:
    \prop_get:NnN \g_restatable_prop { #1 } \l_tmpa_tl
    
    \cs_set:Npn \__restatable_render:nn ##1 ##2 {
      \exp_args:Nc \renewcommand { the ##1 } { \ref{#1} }
      
      \cs_set_eq:NN \label \use_none:n 
      
      \begin{##1} ##2 \end{##1}
    }
    
    \exp_after:wN \__restatable_render:nn \l_tmpa_tl
  \group_end:
 }

\ExplSyntaxOff

\begin{document}
\title{Solution Space Partitioning for Extremal Set Theory}

\author{
  \IEEEauthorblockN{
    Jesse Looney\IEEEauthorrefmark{1}~\orcid{0009-0009-5494-1491},
    Jonah McDonald\IEEEauthorrefmark{1}~\orcid{0009-0009-1627-3197},
    Allison Klingler\IEEEauthorrefmark{1}~\orcid{0009-0002-4850-9562},
    Gloria Wu\IEEEauthorrefmark{1}~\orcid{0009-0008-9861-003X},
    Jonad Pulaj\IEEEauthorrefmark{2}~\orcid{0000-0002-5741-7187},
    Haoze Wu\IEEEauthorrefmark{1}~\orcid{0000-0002-5077-144X}
  }
  \IEEEauthorblockA{\IEEEauthorrefmark{1} Amherst College, Amherst, USA}
  \{jlooney27, jmcdonald27, aklingler27, gwu28, hwu\}@amherst.edu
  \IEEEauthorblockA{\IEEEauthorrefmark{2} Davidson College, Davidson, USA}
  jopulaj@davidson.edu
}

\maketitle
\begin{abstract}
We present a method for partitioning the solution space of statements in extremal set theory. Compared with domain-agnostic partitioning methods like look-ahead, we perform case analysis on the strategies by which a candidate solution can be constructed. We demonstrate that our approach can decompose problems in extremal set theory more effectively than look-ahead. Combining this new partitioning strategy with an exact proof-producing MILP solver, we are able to verify larger finite cases of Chvátal's Conjecture---a long-standing open question in extremal combinatorics\allowbreak{}---\allowbreak{}compared to previous work.

\end{abstract}
\begin{IEEEkeywords}
Extremal Set Theory, Cube and Conquer, Solution Space Partitioning
\end{IEEEkeywords}

\section{Introduction}

Computational methods have increasingly been used to prove mathematical conjectures, particularly in combinatorics and discrete mathematics~\cite{brakensiek_resolution_2022,hales_kepler_2017,heule_schur_2018,heule_solving_2016,kenter_IP-graph-pebbling_2018,giuseppe_IP-counterexamples_2020,parczyk_computer-assisted-combinatorics_2023,pulaj_cutting-planes_2017,subercaseaux_automated_2024}. A common paradigm is to encode statements as constraint satisfaction or optimization problems, and use general-purpose solvers to search for counterexamples or certify correctness. For example, it has been demonstrated that carefully engineered integer programming frameworks can produce machine-checkable proofs for instances of conjectures in extremal combinatorics~\cite{eifler_safe_2022}.

Statements in extremal set theory often concern families of sets over a ground set $[n]$. Even for modest values of $n$, the search space becomes prohibitively large: a family of subsets of $[n]$ is a subset of $2^{[n]}$, so the total number of possible families is $2^{2^n}$, which grows doubly exponentially in $n$. As a result, such problems are typically only tractable for small $n$, with solver performance as the primary bottleneck. Nevertheless, the ability to verify conjectures for larger values of $n$ is valuable: many extremal statements, when false, admit counterexamples at small values of $n$. Extending the verified range therefore both narrows the domain in which counterexamples can exist and increases confidence in the conjecture.

One promising approach to improving solver scalability is parallelization. The state-of-the-art approach for SAT-based theorem proving is Cube and Conquer~\cite{hutchison_cube_2012}, which partitions a problem into thousands to millions of independent subproblems using look-ahead heuristics. While effective in many applications, we observe empirically that Cube and Conquer based on look-ahead heuristics does not effectively decompose problems arising from extremal set theory into substantially easier subproblems.

In this paper, we propose a method for partitioning the solution space of statements in extremal set theory at a semantic level. Instead of branching on variables in the underlying encodings, we partition directly at the level of set families. The key idea is to branch on \emph{strategies} by which a counterexample (i.e., a family of sets) could be constructed. We represent a strategy using representative sets that must be included in the family together with banned subfamilies that must be excluded. We formalize this approach as an abstract calculus that supports 1) repeated partitioning of a collection of candidate families into smaller sub-collections, and 2) pruning of sub-collections that cannot satisfy the statement of interest.  We prove soundness, completeness, termination, and other fundamental proof-theoretic properties of this calculus.

We implement our proof system in Rust, taking as input a problem encoding, a solver oracle, and a ground set. We focus on Chvátal's Conjecture~\cite{chvatal_conjecture_1974}, a long-standing open question in extremal combinatorics, for which the previous best certified computational verification handles ground sets of size $7$~\cite{eifler_safe_2022}. We develop a SAT encoding of Chvátal's Conjecture and demonstrate that our semantic partitioning method decomposes the formula more effectively than look-ahead-based Cube and Conquer. We then apply our partitioning strategy to an existing Integer Linear Programming (ILP) encoding of Chvátal's Conjecture~\cite{eifler_safe_2022} and use the proof-producing MILP solver SCIP~\cite{hojny2025scipoptimizationsuite10} to solve the resulting subproblems. In both the SAT and ILP settings, our partitioning significantly reduces problem difficulty. Moreover, with ILP, we generate and verify machine-checkable proof certificates for the partitions, which constitute a proof of Chvátal's Conjecture for a ground set of size $8$.

In summary, we make the following contributions:
\begin{itemize}
    \item a semantic-level partitioning method for extremal set theory;
    \item a formalization of the method as a calculus, with proofs of soundness, completeness, termination, and related properties;
    \item an oracle-agnostic Rust implementation of our partitioning strategy;
    \item an experimental evaluation using both SAT and ILP oracles, demonstrating the effectiveness of our partitioner; and
    \item a verified proof of Chvátal's Conjecture for $n=8$.
\end{itemize}

\section{Background}
\paragraph{Extremal Combinatorics} Let $n$ be a positive integer. Define the \emph{ground set} $[n]$ to be the set $\{1, 2, \dots, n\}$ and denote the powerset of $[n]$ by $2^{[n]}$. A \textit{family} is a collection of sets---an element of $\Fam = 2^{2^{[n]}}$. A family $d \in \Fam$ is called a \textit{downset} if it is closed under taking subsets---that is, for any $s \subseteq t \in 2^{[n]}$, if $t \in d$, then $s \in d$. An \textit{intersecting} family is one whose members have pairwise nonempty intersections. A special case of an intersecting family is a \textit{star}, a family for which there is some $e \in [n]$ contained in every member of the family. 

\paragraph{Chvátal’s Conjecture} The conjecture we use as a case study is \textit{Chvátal’s Conjecture}. Chvátal’s Conjecture states that, in any downset, there are no intersecting subfamilies strictly larger than the largest star. We think of this conjecture in terms of the existence of a counterexample---an intersecting family that generates a downset whose stars are all smaller than the intersecting family.

\paragraph{Isomorphism of Families} We say two families $f$ and $g$ are isomorphic, and write $f \iso g$, if there exists a permutation $\pi: [n] \to [n]$ such that $\pi(f) = g$ (where $\pi(f) = \set{\pi(s)}{s \in f}$ for any family $f$, and $\pi(s)$ is defined analogously for any set $s$). Isomorphism is an equivalence relation; we denote the equivalence class of a family $f$ under isomorphism by $\eqc f$.

\paragraph{Cube and Conquer}
In \cite{hutchison_cube_2012}, Heule et al. take a new approach to SAT solving by combining CDCL and look-ahead solvers that allows for the solving of larger scale problems. Cube and Conquer, their solution, is also natural to parallelize.

Cube and Conquer uses a partitioning heuristic to select a set of $n$ variables, partitioning the problem into $2^n$ subproblems based on partial assignments. A look-ahead heuristic is used to select the $n$ variables, choosing ones that prune the search space by the greatest amount. The look-ahead solvers choose a variable $x$ and split on it, assigning both true and false values, and use unit propagation to obtain a simplified boolean formula, or a cube.

Cube and Conquer's goal is to generate cubes that simplify the original formula and which can then be solved in parallel as independent cases of the problem. If the difficulty of the original problem is well distributed over the cubes, this approach can result in massive parallel speedup. In the case of Chvátal's Conjecture, however, we found that Cube and Conquer's look-ahead solver produced unbalanced partitions with limited simplification of the original problem. To improve the quality of the partition for problems in extremal set theory, we propose a partitioning method that splits on the inclusion of sets in a family rather than the truth of propositional variables. 

\section{Related Work}

\paragraph{Chvátal's Conjecture} Eifler et al. showed that Chvátal’s conjecture can be reformulated as an integer program, and were able to prove that the conjecture holds for ground sets of size at most $n=7$ using an exact, proof-producing variant of the MILP solver SCIP combined with the certificate checker VIPR \cite{eifler_safe_2022}. Later, an improved version of exact SCIP solved $n=8$ in about 5.2 days \cite{eifler_algorithms_2025}. However, an actual certificate was not generated in this sequental run, and \cite{eifler_algorithms_2025} reports that such proof certificate  would exceed 1\,TB and that this is infeasible for the proof checker VIPR to validate. Plus, generating actual proof would further increase the solver runtime. In this work, we produce a partitioned proof consisting of smaller certificates and verify it with VIPR. Previous mathematical explorations also include \cite{friedgut2016chvatalsconjecturecorrelationinequalities}, which assesses the conjecture using correlation inequalities between Boolean functions on sets. Secondly, \cite{czabarka2017chvatalsconjecturedownsetssmall} proves the conjecture for all downsets where subsets contain at most three elements.

\paragraph{Parallel Solving} There are two common paradigms for parallelizing SAT solvers: portfolio solving, where each thread runs the same problem with a different approach, and divide and conquer, where a problem is partitioned into smaller subproblems that are solved in parallel. Cube and Conquer~\cite{hutchison_cube_2012} is 
a standard approach for parallel SAT solving in the case of theorem proving. Other parallel SAT solvers include Paracooba \cite{paracooba} and PaInleSS \cite{painless}. While there have been some forays into parallelizing ILP solvers \cite{ralphs_parallel-MILP_2018}, we are not aware of any parallelized exact ILP solvers capable of producing proofs.

\paragraph{Symmetry Breaking in SAT and ILP} Symmetry breaking is a method by which to create the subproblems that a divide and conquer parallelization approach requires. \cite{cdcl_sym} introduces symmetry breaking for CDCL-based SAT solving. They utilize a similar method of investigating isomorphisms of a problem that is SAT specific, utilizing a dynamic method to prune the SAT problem as it is being explored. \cite{kirchweger_sat-modulo-sym} also extends CDCL with symmetry breaking and applies it to the problem of generating graphs. \cite{satsuma} utilizes additional symmetry structures within SAT. \cite{margot_pruning_2002} prunes by isomorphisms on ILP problems with constraints of the form $Ax \geq b$, utilizing aspects of group theory applied to an LP solver. We focus on solver- and encoding-agnostic symmetry-breaking specialized to extremal set theory.

\section{Methodology}
We are interested in proving statements about families of sets. To this end, some definitions are in order. We call $\phi$ a \textit{predicate on families} if it has the signature $\phi: \Fam \to \{true, false\}$. For any predicate $\phi$ on families, we define an \textit{existential predicate} $\phi_\exists$ on collections of families by the following rule: for any collection of families $F \in 2^\Fam$, we have $\phi_\exists(F) = true \iff \exists f \in F.\, \phi(f)$. We say $\phi$ is \textit{isomorphism-invariant} if, for any $f, g \in \Fam$ with $f \iso g$, we have $\phi(f) = \phi(g)$. We say $\phi$ is \textit{satisfiable} if there exists $g \in \Fam$ such that $\phi(g) = true$, unsatisfiable otherwise.

\newcommand{\chvatalpred}{\phi^{Chv}}

\begin{example}
We can formulate Chvátal's Conjecture as a family predicate: $\chvatalpred(g) \iff$ "$g$ is an intersecting family which generates a downset such that the two violate Chvátal's Conjecture". Later, we will check using the following methods that $\chvatalpred_\exists(\Fam) = false$ for $n = 8$ (recall $\Fam = 2^{2^{[n]}}$), verifying that there are no counterexamples among these families---Chvátal's Conjecture holds for ground sets up to size $8$.
\end{example}

Given some family predicate $\phi$, we present a method of partitioning $\Fam$, the collection of all families, into subcollections for which the truth of $\phi_\exists$ (e.g. whether the subcollection contains counterexamples to Chvátal's Conjecture) can be checked in parallel. We motivate our method by analogy to Cube and Conquer, which we view as partitioning the collection of all total variable assignments into independent subcollections. Each subcollection is represented by a partial assignment, which contains literals specifying the values of certain variables. The partial assignment stands in for every possible total assignment that could be reached by "completing" the partial assignment by including additional literals. The partition is refined by branching on the truth of some variable, turning one partial assignment into two children that each extend it with a new literal. 

The following functions partially capture these notions and apply them to the realm of families.
\begin{align}
    \ext(f) &= \set{f \cup \{s\}}{s \in 2^{[n]} \setminus f} \\
    \comp(f) &= \set{g \in \Fam}{g \supseteq f} \label{eq:comp-example}
\end{align}
If we think of a family $f$ as a "partial family", then $\comp(f)$ gives all possible "completions" of $f$, achieved by including additional sets, while $\ext(f)$ gives all the incremental "extensions" of $f$ by a single additional set. This pair of functions satisfies the following property:
\begin{equation}
    \label{eq:ext-comp-property}
    \forall f \in \Fam.\, \forall g \in \ext(f).\, \comp(g) \subset \comp(f)
\end{equation}
which says that extending $f$ to obtain $g = f \cup \{s\} \in \ext(f)$ \textit{strictly} reduces the range of possible completions. Just like we extend partial variable assignments with extra literals to focus on particular cases of total assignment, we can extend partial families with extra sets to focus on particular kinds of families.

For instance, observe that $\comp(\emptyset) = \Fam$ contains intersecting families. However, if we consider the extension $\{\emptyset\} \in \ext(\emptyset)$, we find that $\comp(\{\emptyset\})$ contains no intersecting families, because any family containing the empty set cannot be intersecting. If we were searching for intersecting families in $\Fam = \comp(\emptyset)$, we could partition into $\comp(\{\emptyset\})$ and $\Fam \setminus \comp(\{\emptyset\})$, and then immediately eliminate the first collection without exhaustively checking it. However, the notion of "partial families" does not give us any way to denote the exclusion of a particular set, such as $\emptyset$. For this reason, we will introduce a construct equipped with both a partial family and a collection of banned families, analogous to negative literals in a partial variable assignment.

While we make analogy to Cube and Conquer, the crucial difference is that our method reasons directly at the level of families. This allows us to recognize when two apparently different ways of extending a family are in fact isomorphic, which can greatly reduce duplicated work if we know that $\phi$ is isomorphism-invariant. In this section, we first describe a general approach to partitioning $\Fam$ by incrementally extending a partial family. Then we elaborate on an instantiation of this method that employs symmetry-breaking.

\subsection{The Core Proof System}

It turns out that property \ref{eq:ext-comp-property} is all we need to state the most general version of our partitioning scheme; we call any pair of functions $\ext, \comp: \Fam \to 2^\Fam$ an \textit{extension-completion pair} if they satisfy this property.

Given a predicate $\phi$ on families and an extension-completion pair $\ext$ and $\comp$, we define a proof system $\prfsys(\phi, \ext, \comp)$ in the following way. First, a \textit{node} is a pair $N = \node{f}{B}$ of a current family $f$ and a collection $B \in 2^\Fam$ of banned families. When we ban a family, we really want to ban all its completions. Hence, for any collection $B \in 2^\Fam$ of families, we define $\comp(B) = \bigcup_{b \in B} \comp(b)$. Every node $N$ represents a single collection of families, called its \textit{concretization}, given by
\begin{equation}
    N^\# = \comp(f) \setminus \comp(B).
\end{equation}
The concretization of a collection $\Nodes$ of nodes is given by $\Nodes^\# = \bigcup_{N \in \Nodes} N^\#$.

The calculus operates over a \textit{state} consisting of a collection $\Nodes$ of nodes that is updated at each step. There are two inference rules. \branch{} refines the partition by splitting a node into two children based on taking or banning an extension $p$. \prune{} removes nodes whose concretizations contain no families satisfying $\phi$. In either case, we hope to preserve the existence of satisfying families: the $\Nodes$ should contain families satisfying $\phi$ if and only if the same is true for the final $\Nodes$ after some rule applications. 

{
\small
\begin{gather*}
    \branch \\[3pt]
    \frac{
        N = \node{f}{B} \in \Nodes \quad p \in \ext(f) \quad \comp(p) \not\subseteq \comp(B)
    }{
        \Nodes \gets (\Nodes \setminus \{N\}) \cup \{ \node{p}{B}, \node{f}{B \cup \{p\}} \}
    } \\[10pt]
    \prune \\[3pt]
    \frac{
        N \in \Nodes \quad \phi_\exists(N^\#) = false
    }{
        \Nodes \gets \Nodes \setminus \{N\}
    }
\end{gather*}
}

\begin{example}
Let $\phi(f)$ be the statement that $f$ is an intersecting family, and consider the proof system $\prfsys(\phi, \ext, \comp)$ using the definitions of $\ext$ and $\comp$ given above. Let $N = \node{\emptyset}{\emptyset}$. The concretization is $N^\# = \Fam$. If our collection of nodes is $\Nodes = \{N\}$, we can branch on $p = \{\emptyset\}$ to obtain two child nodes $N_1 = \node{\{\emptyset\}}{\emptyset}$ and $N_2 = \node{\emptyset}{\{\{\emptyset\}\}}$ and an updated $\Nodes = \{N_1, N_2\}$. $N_1^\#$ denotes all families containing the empty set, while $N_2^\# = \comp(\emptyset) \setminus \comp(\{\{\emptyset\}\}) = \Fam \setminus \comp(\{\emptyset\})$ denotes all families not containing the empty set. We can easily prove $\neg \phi_\exists(N_1^\#)$ (that is, none of the families in $N_1^\#$ are intersecting families, because they contain the empty set), allowing us to prune $N_1$ and obtain $\Nodes = \{N_2\}$. Note that we have not lost any intersecting families while doing these transformations, a property we will see is guaranteed by the system. 
\end{example}

We prove three highly-general results about $\prfsys$. Theorem \ref{thm:symbreak-preserves-counterexamples} ensures that the inference rules preserve families satisfying $\phi$, so that our search is both sound and complete. Theorem \ref{thm:symbreak-preserves-disjoint} ensures that the "partition" resulting from applications of the inference rules is indeed disjoint. These results follow from induction on Lemmas \ref{lem:children-union} and \ref{lem:children-intersection} proven below (full proofs for these theorems and any other results not proven in the main text can be found in the online appendix).

\begin{lemma}
    \label{lem:children-union}
    Let $N_1 = \node{p}{B}$ and $N_2 = \node{f}{B \cup \{p\}}$ be the child nodes spawned by applying the \branch{} rule to a node $N = \node{f}{B}$ in the proof system $\prfsys(\phi, \ext, \comp)$. Then $N_1^\# \cup N_2^\# = N^\#$.
\end{lemma}
\begin{proof}
    Since $p \in \ext(f)$, we have $\comp(p) \subset \comp(f)$ by Property \ref{eq:ext-comp-property}. It follows that $\comp(p) \cup (\comp(f) \setminus \comp(p)) = \comp(f)$.
    Thus
    \begin{align*}
        N_1^\# \cup N_2^\#
        &= \left[ \comp(p) \setminus \comp(B) \right] \\ & \quad\quad\quad \cup \left[ \comp(f) \setminus \comp(B \cup \{p\}) \right] \\
        &= \left[ \comp(p) \cup (\comp(f) \setminus \comp(p)) \right] \setminus \comp(B) \\
        &= \comp(f) \setminus \comp(B)
        = N^\#.
    \end{align*} 
\end{proof}

\begin{lemma}
    \label{lem:children-intersection}
    Let $N_1 = \node{p}{B}$ and $N_2 = \node{f}{B \cup \{p\}}$ be the child nodes spawned by applying the \branch{} rule to a node $N = \node{f}{B}$ in the proof system $\prfsys(\phi, \ext, \comp)$. Then $N_1^\# \cap N_2^\# = \emptyset$.
\end{lemma}
\begin{proof}
    We have
    \begin{align*}
        N_1^\# \cap N_2^\# &= \left[ \comp(p) \setminus \comp(B) \right]  \\ & \quad\quad\quad \cap \left[ \comp(f) \setminus \comp(B \cup \{p\}) \right] \\
        &= \left[ \comp(p) \cap (\comp(f) \setminus \comp(p)) \right] \setminus \comp(B) \\
        &= \emptyset \setminus \comp(B)
        = \emptyset.
    \end{align*}
\end{proof}

\begin{restatable}{theorem}{thm:symbreak-preserves-counterexamples}
    Suppose $\Nodes_0$ is a collection of nodes, and let $\Nodes_k$ be the resulting collection of nodes after $k \in \nat$ applications of the \branch{} and \prune{} rules of a proof system $\prfsys(\phi, \ext, \comp)$ starting from state $\Nodes_0$. If $\phi$ is a predicate on families, then $\phi_\exists(\mathcal N_k^\#) = \phi_\exists(\mathcal N_0^\#)$.
\end{restatable}

\begin{restatable}{theorem}{thm:symbreak-preserves-disjoint}
    Suppose $\Nodes_0$ is a collection of nodes, and let $\Nodes_k$ be the resulting collection of nodes after $k \in \nat$ applications of the \branch{} and \prune{} rules of a proof system $\prfsys(\phi, \ext, \comp)$ starting from state $\Nodes_0$. If the nodes in $\Nodes_0$ have pairwise disjoint concretizations, then the nodes in $\Nodes_k$ have pairwise disjoint concretizations.
\end{restatable}

The final major result is Theorem \ref{thm:symbreak-terminates}, which says that one cannot apply the inference rules of $\prfsys$ indefinitely. In particular, any decision procedure that iteratively applies the rules until it no longer can is guaranteed to terminate. The prune rule will inevitably run out of targets unless we can keep adding nodes using the branch rule, so we hope that there is some case in which we no longer branch. The following lemma suggests one such scenario.

\begin{restatable}{lemma}{lem:no-branch-on-empty}
    In any proof system $\prfsys(\phi, \ext, \comp)$, the \branch{} rule cannot be applied to nodes with empty concretizations. 
\end{restatable}
\begin{proof}
    Suppose $N = \node{f}{B}$ is a node with $N^\# = \emptyset$. Suppose we wish to branch on $p \in \ext(f)$. We have
    \begin{equation}
        \comp(f) \setminus \comp(B) = N^\# = \emptyset
    \end{equation}
    so $\comp(f) \subseteq \comp(B)$. We also have $\comp(p) \subset \comp(f)$ by Property \ref{eq:ext-comp-property}. It follows that $\comp(p) \subset \comp(B)$, so we cannot apply the branch rule for this arbitrary choice of $p$, because the premise $\comp(p) \not\subseteq \comp(B)$ is not satisfied. Hence, we cannot apply the \branch{} rule.
\end{proof}

Intuitively, as we apply the $\branch$ rule, we more and more precisely specify the range of families represented by the child nodes, shrinking their concretizations until they become empty. Therefore, we want to show that children spawned by the branch rule are smaller, in some sense, than their parent. To this end, we define the \textit{magnitude} of a node $N = \node{f}{B}$ to be $\norm{N} = |\comp(f)||\Fam \setminus \comp(B)|$. The magnitude of a node should not be confused with the cardinality of its concretization, denoted $|N^\#|$. The actual values of nodes' magnitudes are in general only meaningful in comparison to one another. However, we have the important property that a magnitude of zero corresponds to an empty concretization.

\begin{restatable}{lemma}{lem:magnitude-zero-implies-empty}
    For any node $N = \node{f}{B}$, if $\norm{N} = 0$, then $N^\# = \emptyset$.
\end{restatable}
\begin{proof}
    Suppose $\norm{N} = 0$. Then at least one of $|\comp(f)| = 0$ and $|\Fam \setminus \comp(B)| = 0$ must be true. If the former is true, then
    \begin{equation}
        N^\# = \comp(f) \setminus \comp(B) = \emptyset \setminus \comp(B) = \emptyset
    \end{equation}
    If the latter is true, then we must have $\Fam \subseteq \comp(B)$. But then
    \begin{equation}
        \comp(f) \subseteq \Fam \subseteq \comp(B)
    \end{equation}
    so $N^\# = \comp(f) \setminus \comp(B) = \emptyset$. 
\end{proof}

Now we can show that child nodes have strictly smaller magnitudes than their parents. This will immediately allow us to prove that only finite sequences of rule applications are possible in $\prfsys$.

\begin{lemma}
    \label{lem:branch-decreases-magnitude}
    Let $N_1 = \node{p}{B}$ and $N_2 = \node{f}{B \cup \{p\}}$ be the child nodes spawned by applying the \branch{} rule to a node $N = \node{f}{B}$ in the proof system $\prfsys(\phi, \ext, \comp)$. Then $\norm{N_1} < \norm{N}$ and $\norm{N_2} < \norm{N}$.
\end{lemma}
\begin{proof}
    By Lemma \ref{lem:no-branch-on-empty}, since we applied the \branch{} rule, we must have $N^\# \neq \emptyset$. Therefore, by Lemma \ref{lem:magnitude-zero-implies-empty}, we have $\norm{N} \neq 0$, so
    \begin{align}
        |\comp(f)| &> 0 \label{eq:lem-bdm-1}, \\
        |\Fam \setminus \comp(B)| &> 0. \label{eq:lem-bdm-2}
    \end{align}
    
    Since $p \in \ext(f)$, we have $\comp(p) \subset \comp(f)$ by Property \ref{eq:ext-comp-property}, so $|\comp(p)| < |\comp(f)|$. Consequently, using (\ref{eq:lem-bdm-2}), we have 
    \begin{align*}
        \norm{N_1} & = |\comp(p)|\cdot |\Fam \setminus \comp(B)| \\
        & < |\comp(f)|\cdot |\Fam \setminus \comp(B)| = \norm{N}.
    \end{align*}
    
    Moreover, since we applied the \branch{} rule, we must have that $\comp(p) \not\subseteq \comp(B)$. Hence, $|\comp(B \cup \{p\})| = |\comp(B) \cup \comp(p)| > |\comp(B)|$. Therefore,
    \begin{align*}
        |\Fam \setminus \comp(B \cup \{p\})|
        &= |\Fam| - |\comp(B \cup \{p\})| \\
        &< |\Fam| - |\comp(B)| \\
        &= |\Fam \setminus \comp(B)|.
    \end{align*}
    Using (\ref{eq:lem-bdm-1}), we find that
    \begin{align*}
        \norm{N_2}
       &= |\comp(f)|\cdot|\Fam \setminus \comp(B \cup \{p\})|\\
        &< |\comp(f)|\cdot |\Fam \setminus \comp(B)|
        = \norm{N}.
    \end{align*}
\end{proof}

\begin{restatable}{theorem}{thm:symbreak-terminates}
    There are no infinite sequences of rule applications in any proof system $\prfsys(\phi, \ext, \comp)$.
\end{restatable}
\begin{proof}
    (Sketch: See the Appendix for a full proof.) Proof by contradiction. Each \branch{} application reduces the number of high magnitude nodes in favor of lower magnitude nodes, while applying \prune{} can never recover lost magnitude. Inevitably, one runs out of valid (positive magnitude) targets to branch on and subsequently out of nodes to prune.
\end{proof}

\subsection{Breaking Symmetry} 

\newcommand{\extiso}{\ext_{\iso}}
\newcommand{\compiso}{\comp_{\iso}}

With the most general results established, we turn towards a specific instantiation of \prfsys{} with an extension-completion pair that will be suitable for breaking symmetry:

\begin{align}
    \extiso(f) &= \set{f \cup \{s\}}{s \in 2^{[u]} \setminus f} \\
    \compiso(f) &= \set{g \in \Fam}{\exists f' \in [f].\, g \supseteq f'} \label{eq:compiso}
\end{align}

The definition of $\extiso$ the same as the $\ext$ we defined in an earlier example. In contrast, $\compiso$ explicitly references isomorphism, defining the completion of a family $f$ to consist of any family that contains $f$ \textit{or a family isomorphic to $f$} (compare the definition of $\compiso$ with Equation \ref{eq:comp-example}). What this means for the execution of \prfsys{} is that when we ban the completion of a family, we are able to eliminate many families with entirely different compositions that are nonetheless isomorphic (and thus redundant, when the family-predicate $\phi$ is isomorphism-invariant).

\begin{example}
Let $N = \node{\emptyset}{\emptyset}$. Branching gives $N_1 = \node{\{\{1\}\}}{\emptyset}$ and $N_2 = \node{\emptyset}{\{\{\{1\}\}\}}$. Note that $[\{\{1\}\}] = \set{\{\{i\}\}}{i \in [n]}$, so $\compiso(\{\{1\}\})$ includes all families that contain \textit{any} singleton set. Hence, $N_2^\# = \compiso(\emptyset) \setminus \compiso(\{\{1\}\})$ denotes all families that do not contain a singleton set. Node $N_2$ now does not admit branching on, say, $\{\{2\}\}$, since $\compiso(\{\{2\}\}) = \compiso(\{\{1\}\})$ (violating the $\comp(p) \not\subseteq \comp(B)$ premise). Thus, banning $\{\{1\}\}$ also banned $\{\{2\}\}$. This prevents us from duplicating work, as we already consider the isomorphic case in $N_1$.
\end{example}

In this subsection, we discuss the implications of specifying this extension-completion pair on the properties of \prfsys{}. First, we introduce two quick but important properties of $\compiso$, and we apply them to show that $\extiso$ and $\compiso$ do form an extension-completion pair as we have suggested.

\begin{restatable}{lemma}{lem:f-in-compiso-f}
    For any family $f \in \Fam$, we have $f \in \compiso(f)$.
\end{restatable}
\begin{proof}
    This follows from $f \in [f]$ and $f \supseteq f$. 
\end{proof}

\begin{restatable}{lemma}{lem:compiso-f-in-compiso-g}
    Suppose $f, g \in \Fam$ with $f \in \compiso(g)$. Then $\compiso(f) \subseteq \compiso(g)$.
\end{restatable}
\begin{proof}
    We know from $f \in \compiso(g)$ that $f \supseteq g'$ for some $g' \in [g]$ (so there is a permutation $\pi_g$ that sends $g$ to $g'$). Suppose $h \in \compiso(f)$. Then $h \supseteq f'$ for some $f' \in [f]$ (so there is a permutation $\pi_f$ that sends $f$ to $f'$). Now we have 
    \begin{equation}
        h \supseteq f' = \pi_f(f) \supseteq \pi_f(g') =\pi_f(\pi_g(g)).
    \end{equation}
    Moreover, $\pi_f \circ \pi_g$ is a permutation, so $\pi_f(\pi_g(g)) \in [g]$, and thus $h \in \compiso(g)$. 
\end{proof}

\begin{theorem}
    \label{thm:extiso-compiso-satisfy-ext-comp}
    The functions $\extiso$ and $\compiso$ are an extension-completion pair.
\end{theorem}
\begin{proof}
    Suppose $f \in \Fam$ and $g \in \extiso(f)$. Then $g = f \cup \{s\}$ for some $s \in 2^{[n]} \setminus f$. In particular, $g \supsetneq f$, so $g \in \compiso(f)$. By Lemma \ref{lem:compiso-f-in-compiso-g}, that means $\compiso(g) \subseteq \compiso(f)$.

    To show that $\compiso(g) \neq \compiso(f)$, we will note that $f \in \compiso(f)$ by Lemma \ref{lem:f-in-compiso-f}, and show that $f \notin \compiso(g)$. Suppose $f \in \compiso(g)$. Then $f \supseteq g'$ for some $g' \in [g]$. But
    \begin{equation}
        |g'| = |g| = |f \cup \{s\}| > |f|
    \end{equation}
    because $s \notin f$. Therefore, $|f| \geq |g'| > |f|$, a contradiction. 
\end{proof}

Now that we have Theorem \ref{thm:extiso-compiso-satisfy-ext-comp}, we are justified in considering instantiations of \prfsys{} of the form $\prfsys(\phi, \extiso, \compiso)$. It is expected that such instantiations will additionally require that $\phi$ be isomorphism-invariant, but this constraint is not necessary for the results we develop here.

The main result we aim for is a strengthening of Lemma \ref{lem:branch-decreases-magnitude}: when breaking symmetry, branching not only decreases the (rather contrived) \textit{magnitude} of a node, but directly decreases the \textit{cardinality} of a node's concretization. All we require before proving this is a small lemma:

\begin{lemma}
    \label{lem:no-branch-if-f-banned}
    In a proof system $\prfsys(\phi, \extiso, \compiso)$, the \branch{} rule cannot be applied to nodes $N = \node{f}{B}$ where $f \in \compiso(B)$.
\end{lemma}
\begin{proof}
    Suppose $f \in \compiso(B)$. Then $f \in \compiso(b)$ for some $b \in B$. By Lemma \ref{lem:compiso-f-in-compiso-g}, that means $\compiso(f) \subseteq \compiso(b) \subseteq \compiso(B)$. Hence $N^\# = \compiso(f) \setminus \compiso(B) = \emptyset$. By Lemma \ref{lem:no-branch-on-empty}, we cannot apply the \branch{} rule to $N$. 
\end{proof}

\begin{theorem}
    \label{thm:branch-decreases-card}
    Let $N_1 = \node{p}{B}$ and $N_2 = \node{f}{B \cup \{p\}}$ be nodes spawned by applying the \branch{} rule to a node $N = \node{f}{B}$ in a proof system $\prfsys(\phi, \extiso, \compiso)$. Then $|N_1^\#|, |N_2^\#| < |N^\#|$.
\end{theorem}
\begin{proof}
    Lemmas \ref{lem:children-union} and \ref{lem:children-intersection} together tell us that $|N_1^\#| + |N_2^\#| = |N^\#|$. So we need only show $0 < |N_1^\#| < |N^\#|$. By Lemma \ref{lem:no-branch-if-f-banned}, we have $f \notin \compiso(B)$. We also have $f \in \compiso(f)$ by Lemma \ref{lem:f-in-compiso-f}, so $f \in \compiso(f) \setminus \compiso(B) = N^\#$. But $f \notin \compiso(p)$ (or else $\compiso(f) \subseteq \compiso(p)$ by Lemma \ref{lem:compiso-f-in-compiso-g}, violating Property \ref{eq:ext-comp-property}), so $f \notin \compiso(p) \setminus \compiso(B) = N_1^\#$. Hence, $N_1^\# \subset N^\#$, which means $|N_1^\#| < |N^\#|$. Moreover, since we applied the \branch{} rule, we must have $\compiso(p) \not\subseteq \compiso(B)$, so $N_1^\# = \compiso(p) \setminus \compiso(B) \neq \emptyset$, which implies $|N_1^\#| > 0$. 
\end{proof}

\subsection{Pruning with SAT}
Implementing a decision procedure for $\prfsys(\phi, \extiso, \compiso)$ that is able to apply the \prune{} rule requires a decision procedure for determining $\neg\phi_\exists(N^\#)$. Below, we outline how to leverage automated reasoning tools to do this.

In the most general terms, any constrained system of binary variables generates a family predicate. Suppose we have a set $V = \set{v_s}{s \in 2^{[n]}}$ of binary variables $v_s \in \{true, false\}$. For any family $g \in \Fam$, we can define a unique assignment $A_g: V \to \{true, false\}$ by the rule $A_g(v_s) = true$ iff $s \in g$. Hence, given any constraint $K = \{A_1, \dots, A_k\}$ of allowed (e.g. satisfying) assignments, we obtain a family predicate by defining $\phi(g) = true$ iff $A_g \in K$.

Many constraint languages admit interpretation in this way. For example, any propositional formula, together with a mapping of some of its variables to sets, generates a family predicate under the natural constraint that $K$ is the set of assignments to those variables (i.e. partial assignments) that can be extended into satisfying assignments. As another example, consider the ILP problem $P_{red}(n)$ in \cite{eifler_safe_2022}. Through our lens, $V$ is the set of the $y_S$ variables, and $K$ is the set of assignments to those variables such that the ILP is feasible with optimal objective value zero. This constrained system generates a family predicate, which the authors showed is equisatisfiable to $\chvatalpred$.

Suppose $\phi$ is a predicate on families. For any node $N = \node{f}{B}$ in the proof system $\prfsys(\phi, \extiso, \compiso)$, we define the \textit{restriction of $\phi$ corresponding to $N$} by
\begin{align}
    \phi_N(g) &= \phi(g) \land f \subseteq g \land g \notin \compiso(B)
\end{align}
That is, $\phi_N$ is the same predicate as $\phi$ except that satisfying families must additionally (1) contain $f$ and (2) must not be banned. We thus limit our focus to satisfying families contained in $N^\#$. Even further, we break symmetry by restricting to families that actually contain the representative $f$ itself (not just isomorphic copies), making it feasible to automatically deduce the satisfiability of $\phi_N$ in practice.

For example, if $\psi$ is a propositional formula containing variables $v_s$ for all $s \in 2^{[n]}$, we can define
\begin{align}
    \psi_N &= \psi \land \alpha_f \land \beta_B \\
    \alpha_f &= \bigwedge_{s \in f} v_s \\
    \beta_B &= \bigwedge_{b \in B} \bigwedge_{b' \in [b]} \bigvee_{s \in b'} \neg v_s
\end{align}
If the predicate generated by $\psi$ over the $v_s$ is $\phi$, then the predicate generated by $\psi_N$ must be $\phi_N$. The $\alpha_f$ constraints enforce that a satisfying assignment represents a superset of $f$, while the $\beta_B$ constraints require that it is not in $\compiso(B)$ (by ensuring it is not a superset of anything isomorphic to some $b \in B$).

Quite analogously, if $\omega$ is an ILP containing variables $v_s \in \{0, 1\}$ for all $s \in 2^{[n]}$, we can add additional constraints on those variables:
\begin{align}
    v_s = 1 & \quad s \in f \label{ilp1} \\
    \sum_{s \in b'} v_s \leq |b'| - 1 & \quad b' \in [b],\, b \in B \label{ilp2}
\end{align}
If $\omega$ and the $v_s$, together with a requirement on the feasibility or optimal value of $\omega$, generate a family predicate $\phi$, then the modified problem with the new constraints generates $\phi_N$. Analogous encodings could be derived for any other constraint language. Note that the banned family constraints ($\beta_B$ in the SAT encoding) can in principal result in an exponential blowup of the encoding size. While this is bound to happen at some point, we have not yet observed in practice the subproblem encoding incurring an overhead that outweights the benefit of solution space reduction.

Now that we know we can automatically deduce the satisfiability of the predicate $\phi_N$ using a constraint language and solver of our choice, we wish to show that proving $\phi_N$ unsatisfiable is enough to justify pruning the node $N$.

\begin{restatable}{lemma}{lem:banned-formula}
    Suppose $B$ is a collection of families. Suppose $g_1 \iso g_2$ are families. If $g_1 \in \compiso(B)$, then $g_2 \in \compiso(B)$.
\end{restatable}
\begin{proof}
    Suppose $g_1 \in \compiso(B)$. Then $g_1$ contains some $b' \iso b$ for some $b \in B$. We are guaranteed an isomorphism $\pi(g_1) = g_2$. Therefore, $g_2 = \pi(g_1) \supseteq \pi(b') \iso b$, so $g_2 \in \compiso(B)$.
\end{proof}

\begin{theorem}
    \label{thm:prune-formula}
    Suppose $\phi$ is an isomorphism-invariant predicate on families. Suppose $N = \node{f}{B}$ is a node in the proof system $\prfsys(\phi, \extiso, \compiso)$. Then $\phi_\exists(N^\#)$ if and only if $\phi_N$ is satisfiable.
\end{theorem}
\begin{proof}
    $(\implies)$ Suppose $\phi_\exists(N^\#)$. Then there is a family $g' \in N^\# = \compiso(f) \setminus \compiso(B)$ that satisfies $\phi(g')$. Since $g' \in \compiso(f)$, we have $g' \supseteq f'$ for some $f' \in [f]$ (so there exists a permutation $\pi_f$ that sends $f'$ to $f$). Let $g = \pi_f(g') \in [g']$, so that $g = \pi_f(g') \supseteq \pi_f(f') = f$. Since $g' \notin \compiso(B)$, we have $g \notin \compiso(B)$ (contrapositive of Lemma \ref{lem:banned-formula}). Finally, since $\phi$ is isomorphism-invariant, we must have $\phi(g) =\phi(g') = true$. Therefore, $\phi_N(g) = true$, so $\phi_N$ is satisfiable.

    $(\impliedby)$
    Suppose there is some $g \in \Fam$ satisfying $\phi_N(g)$. Then $\phi(g) = true$ and $f \subseteq g$ and $g \notin \compiso(B)$. Therefore, $g \in N^\#$ and $\phi(g) = true$, so $\phi_\exists(N^\#)$.
\end{proof}

For a given node $N = \node{f}{B}$, it is possible to restrict $\phi_N$ even further with \textit{symmetry constraints}. Choose any $p \in \ext(f)$, and define
\begin{align}
    \phi_N^p(g) = \phi_N(g) \land (g \in \compiso(p) \implies p \subseteq g)
\end{align}
That is, we add the constraint that satisfying families must contain $p$, or else not contain anything isomorphic to $p$. The idea is, if we can satisfy $\phi$ with a family containing $p' \iso p$, we can instead satisfy $\phi$ (assuming isomorphism-invariance) with an isomorphic family that contains $p$ itself. The added constraint forces us to "look for" solutions containing $p$ "before" considering those containing isomorphic families, limiting the search space further. This constraint could be encoded in SAT as follows:
\begin{align}
    \psi_N^p &= \psi_N \land \sigma_p \\
    \sigma_p &= \left( \bigwedge_{s \in p \setminus f} v_s \right) \lor \left( \bigwedge_{p' \in [p]} \bigvee_{t \in p'} \neg v_t \right) \\
        &= \bigwedge_{s \in p \setminus f} \bigwedge_{p' \in [p]} \left( v_s \lor \bigvee_{t \in p'} \neg v_t \right)
\end{align}
We evaluate the effect of adding these constraints empirically in Section~\ref{sec:milp-experiments}; they did not improve solving performance in our experiments, so we do not impose them in the main results we report.

\begin{restatable}{theorem}{thm:prune-formula-sym}
    Suppose $\phi$ is an isomorphism-invariant predicate on families. Suppose $N = \node{f}{B}$ is a node in the proof system $\prfsys(\phi, \extiso, \compiso)$, and fix some $p \in \ext(f)$. Then $\phi_\exists(N^\#)$ if and only if $\phi_N^p$ is satisfiable.
\end{restatable}

\section{Implementation}

We implemented a decision procedure for \prfsys{} as a Rust library.\footnote{Source code available at \url{https://github.com/jesselooney/symbreak}.} Users can implement the \texttt{NodeOracle} trait (interface) by providing methods for serializing and solving subproblems given input nodes (e.g. writing a DIMACS file and invoking a SAT solver on that file). The partitioning procedure uses a given \texttt{NodeOracle} to determine whether individual nodes can be pruned. The user needs to provide an implementation of the oracle, which determines the family-predicate $\phi$ being tested. Our prototype also supports a dynamic timeout system for the oracle based on the depth of the node.

Our partitioner, \texttt{symbreak}, branches on nodes in which the \texttt{NodeOracle} cannot determine the existence of a counterexample. We use \texttt{nauty} to check isomorphism of families in order to build the equivalence classes \cite{nauty}. Nodes are branched on and tested for prunability in parallel, and the partitioner logs its progress as it executes. We implemented a verifier script in Python that inspects the logs and the set of terminal subproblems to ensure that all nodes are accounted for.

To apply \branch{} rule, one needs to choose an extension $p$ to branch on. Our implementation heuristically chooses the extension with the smallest equivalence class. We experimented with alternative heuristics (e.g. the largest equivalence class), but did not find one that was consistently better than the others.

\subsection{Case Study: Chvátal's Conjecture}
As a test case for our partitioner, we used it to verify Chvátal's Conjecture for finite ground sets. We developed an implementation of \texttt{NodeOracle} that encodes the conjecture as a boolean satisfiability problem using RustSAT \cite{Jabs2025RustsatLibrarySat} and determines prunability with Kissat \cite{biere_cadical_kissat_2024}. We developed a novel SAT encoding of the conjecture based on the ILP encoding in \cite{eifler_safe_2022}. For each set $s \in 2^{[n]}$, we define propositional variables $x_s$ and $y_s$. The truth of the former is interpreted as meaning that $s$ is a member of a family $X$, the latter as meaning that $s$ is a member of a family $Y$. We impose the following constraints (viewed as clauses) on these variables:
\begin{align}
    \{\neg y_s, \neg y_t\} \quad & s,t \in 2^{[n]}, s \cap t = \emptyset \label{c1} \\
    \{\neg y_t, x_s\} \quad & s,t \in 2^{[n]}, s \subseteq t \label{c2} \\
    \sum_{s \in 2^{[n]}} y_s > \sum_{s \in 2^{[n]}, e \in s} x_s \quad & e \in [n] \label{c3} \\[10pt]
    \{\neg y_s\} \quad & s \in 2^{[n]}, |s| \in \{1,2\} \label{c4} \\
    \{x_s\} \quad & s \in 2^{[n]}, |s| = 1 \label{c5} \\
    \{\neg x_s\} \cup \set{y_t}{t \in 2^{[n]}, t \supseteq s} \quad & s \in 2^{[n]}, |s| > 1 \label{c6} \\
    \begin{matrix}
        \set{y_t}{t \in 2^{[n]}, t \cap s = \emptyset} \\
        \cup \{\neg x_s, y_s\}
    \end{matrix} \quad & s \in 2^{[n]} \label{c7}
\end{align}

Constraints \ref{c6} and \ref{c7} are justified by the following theorem.

\begin{theorem}
    If there exists a counterexample to Chvátal's Conjecture, then there exists a counterexample (over the same ground set) in which the downset is minimal and the intersecting family is maximal.
\end{theorem}
\begin{proof}
    Suppose there is a counterexample to Chvátal's Conjecture, so there is a downset $X$ containing an intersecting family larger than every contained star. There may be multiple such intersecting families, but we are free to let $Y$ be a maximal one. Now we can let $X'$ be the minimal downset generated by $Y$. Since $Y \subseteq X$, and $X$ is a downset, we have $X' \subseteq X$. Clearly, then, the largest stars in $X'$ are no larger than those in $X$. Thus $X'$ and $Y$ still constitute a counterexample to Chvátal's Conjecture.
\end{proof}
Note that one of the constraints is a cardinality constraint, which we encode using a totalizer \cite{bailleux_totalizer_2003}. Together, constraints \ref{c1}--\ref{c3} assert that $Y$ is a counterexample to Chvátal's Conjecture: $Y$ is an intersecting family which generates a downset $X$ in which all the stars are smaller than $Y$. These constraints alone would generate (over the variables $y_s$) a family predicate equivalent to $\chvatalpred$. The remaining constraints limit the search space further, while ensuring that any counterexample they might exclude would have a counterpart they permit (the interpretation of each constraint is given in Table \ref{tab:chvatal-sat-interpretations}). The resulting formula thus generates a predicate $\phi^{SAT}$ that is equisatisfiable to $\chvatalpred$. Moreover, $\phi^{SAT}$ is isomorphism-invariant, since each constraint depends only on the structure of $X$ and $Y$, not the particular names of the elements they contain. Hence, Theorem \ref{thm:prune-formula} justifies using the satisfiability of $\phi^{SAT}_N$ (implemented as $\psi_N$) to determine whether a node $N$ can be pruned in the proof system $\prfsys(\phi^{SAT}, \extiso, \compiso)$. We did not add the symmetry constraints $\sigma_p$ in the experiment since we found that their effect is mixed.

\begin{table}[t]
    \centering
    \caption{Interpretation of each constraint in our SAT encoding.}
    \label{tab:chvatal-sat-interpretations}
    \setlength{\tabcolsep}{5pt}
    \begin{tabular}{rp{0.65\linewidth}}
        \toprule
        Constraints & Interpretation \\
        \midrule
        \ref{c1} & $Y$ is an intersecting family. \\
        \ref{c2} & $X$ contains the downset generated by $Y$. \\
        \ref{c3} & $Y$ has greater cardinality than that of any star in $X$. \\
        \ref{c4} & See constraints (7f) in \cite{eifler_safe_2022}. \\
        \ref{c5} & See constraints (7g) in \cite{eifler_safe_2022}. \\
        \ref{c6} & $X$ is a \textit{minimal} downset containing $Y$; $X$ contains only sets with supersets in $Y$. \\
        \ref{c7} & $Y$ is a \textit{maximal} intersecting family in $X$; any set in $X \setminus Y$ must be disjoint with something in $Y$. \\
        \bottomrule
    \end{tabular}
\end{table}

Additionally, we developed a separate implementation of \texttt{NodeOracle} that encodes the conjecture as an integer linear program using the PySCIPOpt \cite{maher_pyscipopt_2016} API for exact rational solving with SCIP \cite{hojny2025scipoptimizationsuite10}. We use the same constraints as in $P_{red}$ of \cite{eifler_safe_2022}, with the exception of (7h) (see below), to specify the base problem. We require that the optimal objective value of the ILP is zero, thus generating a family predicate $\phi^{ILP}$ equisatisfiable to $\chvatalpred$ (as shown by \cite{eifler_safe_2022}). Because we left out constraint (7h), the only constraint that cares about the particular names of the elements of $X$ and $Y$, the generated predicate is isomorphism-invariant. Therefore, Theorem \ref{thm:prune-formula} ensures the correctness of pruning by adding the following constraints for a node $N = \node{f}{B}$:
\begin{align}
    y_s = 1 & \quad s \in f \\
    \sum_{s \in b'} y_s \leq |b'| - 1 & \quad b' \in [b],\, b \in B
\end{align}
and checking whether the resulting ILP has optimal value zero. When partitioning with this oracle, we first branch on the extension $\{\{1,2,3,4\}\}$ to partially recover the symmetry breaking that (7h) provides in the serial encoding (taking this family implies (7h)), and fall back to the default branching heuristic described above for all subsequent branches.

\begin{table*}[t]
    \centering
    \begin{threeparttable}
        \caption{Partition sizes, generation time, and total time to solve resulting instances, for \texttt{symbreak} (oracle timeouts 1s and 3s) and \texttt{march\_cu} (default config and depths 15 and 20) on $n=7$. Generation wall times are given only for the parallel partitioners.}
        \label{tab:sat-partitions}
        \setlength{\tabcolsep}{5pt}
        \begin{tabular}{lS[table-format=4.0]S[table-format=3.1]S[table-format=4.1]S[table-format=5.1]S[table-format=4.1]}
            \toprule
            \multirow{2}{*}{partitioner} &
                {\multirow{2}{*}{partition size}} &
                \multicolumn{2}{c}{generation time (s)} &
                {total solve time (s)} & {max solve time (s)} \\
            & & {wall} & {cpu} & {cpu} & {cpu}\\
            \midrule
            \texttt{symbreak} (1s)   &   102 & 411.1 & 6138.9 &   622.3           &   47.5 \\
            \texttt{symbreak} (3s)   &    46 & 120.8 & 2621.7 &   396.3           &   34.6 \\
            \texttt{march\_cu}       &  1570 & {---} &    9.0 &  8650.0 \tnote{*} &   {TO} \\
            \texttt{march\_cu} (d15) &  2821 & {---} &   21.7 & 39127.3           & 2102.6 \\
            \texttt{march\_cu} (d20) & 44065 & {---} &  356.9 & 56843.9           &  276.6 \\
            \bottomrule
        \end{tabular}
        \begin{tablenotes}
            \item [*] Does not include four instances that timed out after one hour.
        \end{tablenotes}
    \end{threeparttable}
\end{table*}

\begin{figure*}[h!]
\centering
    \begin{minipage}{0.3\textwidth}
    \centering
        \includegraphics[width=\textwidth]{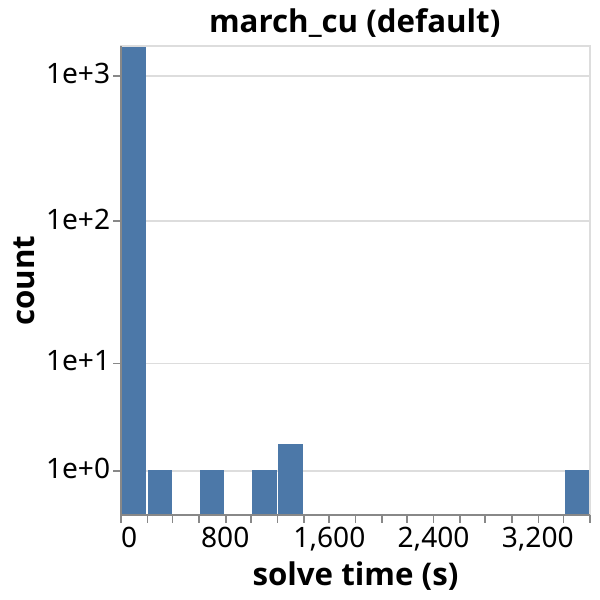}
        \subcaption{}
        \label{fig:march-partition-histplot}
    \end{minipage}
    \begin{minipage}{0.3\textwidth}
    \centering
        \includegraphics[width=\textwidth]{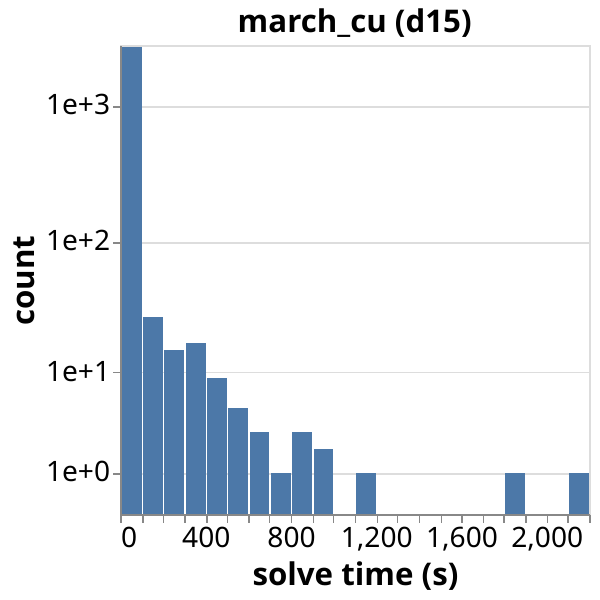}
        \subcaption{}
        \label{fig:march-d15-partition-histplot}
    \end{minipage}
    \begin{minipage}{0.3\textwidth}
    \centering
        \includegraphics[width=\textwidth]{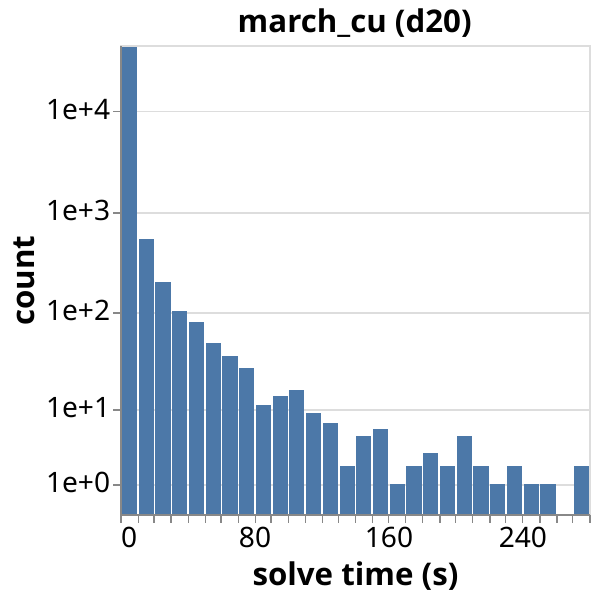}
        \subcaption{}
        \label{fig:march-d20-partition-histplot}
    \end{minipage}
    \begin{minipage}{0.3\textwidth}
    \centering
        \includegraphics[width=\textwidth]{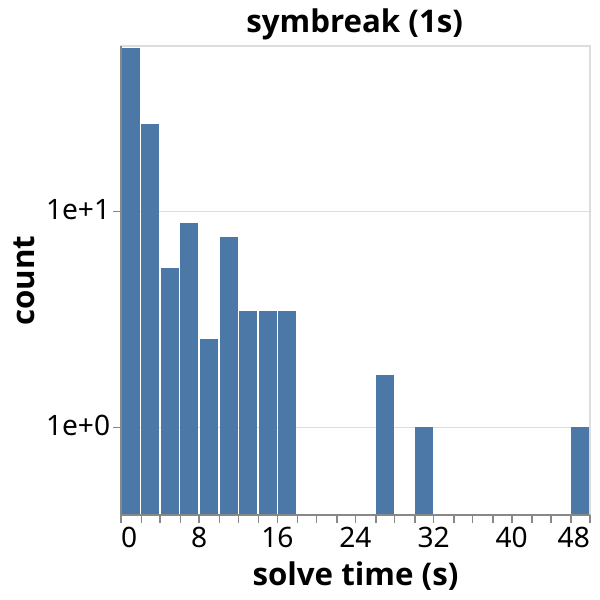}
        \subcaption{}
        \label{fig:sat-t1-partition-histplot}
    \end{minipage}
    \begin{minipage}{0.3\textwidth}
    \centering
        \includegraphics[width=\textwidth]{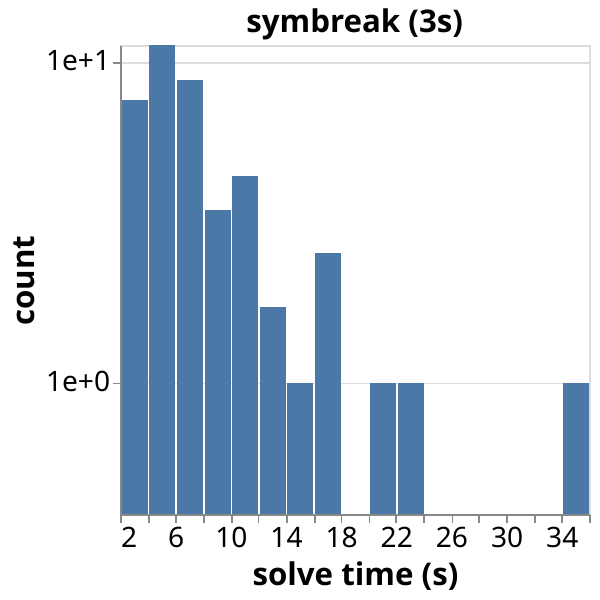}
        \subcaption{}
        \label{fig:sat-t3-partition-histplot}
    \end{minipage}
    \caption{Distribution of solve times for subproblems of $n=7$ generated by \texttt{march\_cu} (a) default config (excluding four timed-out instances), (b) depth 15, (c) depth 20; and SAT-based \texttt{symbreak} with oracle timeouts (d) 1s and (e) 3s. The $y$ axis is log-scaled in each case.}
\end{figure*}

\section{Experiments}
We evaluated the instantiations of our partitioner described in the previous section, comparing them to state-of-the-art techniques applied to the same problem of solving finite cases of Chvátal's Conjecture. In the SAT world, the dominant paradigm is Cube and Conquer---in particular, look-ahead-based partitioning. We compare our partitioner to the look-ahead solver \texttt{march\_cu} \cite{heule_cnc_2026}.

For ILP, we instead run our partitioner in solving mode and compare against a state-of-the-art exact solver, SCIP. In doing this, we obtain the first fully verified solution of the case $n = 8$. We detail these experiments in the following subsections. Unless noted otherwise, all computations were run on a computing cluster consisting of Dell PowerEdge R6525 rack servers with 2.6-GHz AMD CPUs.

\subsection{SAT Partitioning: Comparison with Look-Ahead}
We compared our partitioner with \texttt{march\_cu}, the state-of-the-art partitioner used in Cube and Conquer, on a SAT instance of the Chvátal problem for $n = 7$. We ran two configurations of our partitioner, each with a different oracle timeout (1s and 3s), both with 16 threads and a maximum depth of 20. We ran three configurations of \texttt{march\_cu}, running the default configuration as well forcing it to partition to depth 15 or 20 with the \texttt{-d} flag. Table \ref{tab:sat-partitions} reports the number of subproblems generated by each partitioner and the time each took to generate. Each subproblem was then solved using Kissat 4.0.4 with the \texttt{{-}{-}unsat} flag with a one-hour time limit. For each partitioner, we report in Table \ref{tab:sat-partitions} the total CPU time taken to solve all subproblems that finished within the time limit.

\begin{table}[t!]
    \caption{ILP dynamic splitting run times on 52 cores with tuned SCIP settings. Each configuration $\{t, r\}$ uses an initial per-node timeout of $t$ seconds up to depth 10, after which the timeout is multiplied by $r$ for each subsequent depth. Wall and total times are in H:M:S format; leaves is the number of leaf subproblems in the final partition.}
    \label{tab:milp-times-tuned}
    \setlength{\tabcolsep}{4pt}
    \begin{center}
        \begin{tabular}{l rrr rrr}
            \toprule
            & \multicolumn{3}{c}{$n = 7$} & \multicolumn{3}{c}{$n = 8$} \\
            \cmidrule(lr){2-4} \cmidrule(lr){5-7}
            Config. & {Wall} & {Total} & {Leaves} & {Wall} & {Total} & {Leaves} \\
            \midrule
            Serial & \text{0:00:24} & {---} & {---} & \text{72:26:02} & {---} & {---} \\
            \midrule
            \{5, 2\}  & \text{0:01:00} & \text{0:01:19} & 11 & \text{2:33:11} & \text{27:26:58} & 197 \\
            \{5, 5\}  & \text{0:01:01} & \text{0:01:19} & 11 & \text{2:15:59} & \text{6:45:03}  & 76 \\
            \{60, 2\} & \text{0:00:56} & \text{0:00:56} & 1  & \text{4:11:04} & \text{12:40:15} & 53 \\
            \{60, 5\} & \text{0:00:56} & \text{0:00:56} & 1  & \text{1:53:49} & \text{5:08:38}  & 33 \\
            \bottomrule
        \end{tabular}
    \end{center}
\end{table}

\begin{table}[t!]
    \caption{ILP dynamic splitting run times on 52 cores with default SCIP settings. The $n = 7$ runs use switch depth $s = 10$ and the $n = 8$ runs use $s = 16$ (the initial timeout $t$ applies up to depth $s$ and is multiplied by $r$ for each subsequent depth). Wall and total times are in H:M:S format; leaves is the number of leaf subproblems in the final partition. The serial $n = 8$ run did not terminate within seven days.}
    \label{tab:milp-times-default}
    \setlength{\tabcolsep}{4pt}
    \begin{center}
        \begin{tabular}{l rrr rrr}
            \toprule
            & \multicolumn{3}{c}{$n = 7$} & \multicolumn{3}{c}{$n = 8$} \\
            \cmidrule(lr){2-4} \cmidrule(lr){5-7}
            Config. & {Wall} & {Total} & {Leaves} & {Wall} & {Total} & {Leaves} \\
            \midrule
            Serial & \text{0:05:24} & {---} & {---} & \text{$>$7 days} & {---} & {---} \\
            \midrule
            \{5, 7\}  & \text{0:03:03} & \text{0:08:20} & 35 & \text{19:15:05} & \text{340:26:21} & 1679 \\
            \{10, 6\} & \text{0:04:23} & \text{0:08:53} & 26 & \text{19:28:20} & \text{321:43:16} & 1214 \\
            \{20, 5\} & \text{0:05:12} & \text{0:10:36} & 21 & \text{19:21:27} & \text{322:30:35} & 870 \\
            \{60, 4\} & \text{0:06:03} & \text{0:07:12} & 6  & \text{20:24:28} & \text{303:47:54} & 570 \\
            \bottomrule
        \end{tabular}
    \end{center}
\end{table}

Note that we spent significant effort trying different SAT encodings of Chvátal's Conjecture by tuning the cardinality encodings and adding or removing known encodings of various lemmas that strengthen our assumptions for the conjecture. We also attempted to tune Kissat as in \cite{wu2025cubingtuning} and to use a solver that was more aware of cardinality constraints \cite{klausesclauses}. None of these methods helped significantly, motivating our current approach.

Figures \ref{fig:march-partition-histplot}--\ref{fig:sat-t3-partition-histplot} show the distribution of solve time over the subproblems generated by each partitioner configuration. The distributions share the same right-skew, with more easy subproblems and a few harder ones, though those from our partitioner have a smaller range of solve times. Only the depth 20 configuration of \texttt{march\_cu} approaches our narrow range, but does so at the cost of nearly a thousand times more instances. The data suggest that \texttt{march\_cu} is decomposing the original problem less effectively than \texttt{symbreak}, carving off many more trivial subproblems in the process. We expect \texttt{symbreak}'s advantage is due to its better internal representation of the problem structure and that it is guaranteed to strictly reduce the solution space at each branch (Theorem \ref{thm:branch-decreases-card}).

\subsection{Solving $n=8$ with ILP and Dynamic Partitioning}
\label{sec:milp-experiments}
We found that SAT solvers tend to be significantly slower than ILP solvers for verifying the same value of $n$. Therefore, to tackle Chvátal's Conjecture for larger $n$, we used an existing ILP encoding and the exact solving mode of SCIP 10.0.3. Instead of pre-determining the depth to which we perform partitioning, we employ a dynamic divide-and-conquer scheme. Each configuration is defined as $\{t, r\}$: the initial per-node oracle timeout of $t$ seconds is used for all nodes up to a switch depth $s$, after which the timeout is multiplied by $r$ for each subsequent depth. We ran each configuration using 52 cores.

Except where noted otherwise, we ran SCIP with conflict analysis disabled (\texttt{conflict/enable = false}) and Gomory cuts disabled (\texttt{separating/gomory/freq = -1}). We found that disabling these features is significantly better in both the sequential and the divide-and-conquer setting.

\begin{figure}
    \centering
    \includegraphics[width=0.35\textwidth]{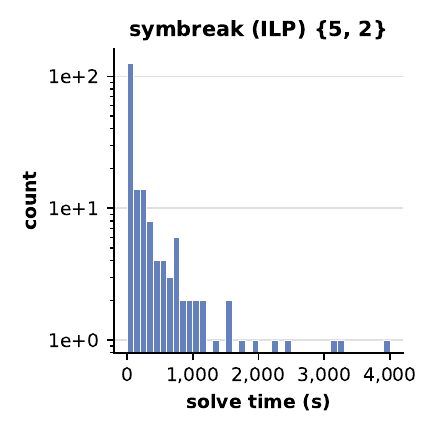}
    \caption{Distribution of solve times for leaf subproblems in the \{5, 2\} configuration for dynamic splitting with ILP and $n=8$.}
    \label{fig:milp-8052-histogram}
\end{figure}

Table~\ref{tab:milp-times-tuned} reports, for each configuration with the tuned SCIP settings and switch depth $s = 10$, the wall-clock time, the total time (the sum of all per-node solving and branching times), and the number of leaf subproblems in the final partition. With these settings, serial exact SCIP verifies $n = 7$ in 24 seconds; in the $\{60, 2\}$ and $\{60, 5\}$ configurations, the root problem for $n = 7$ is solved within the initial timeout, so no partitioning takes place and the wall time is simply the root solve time. The root problem takes longer than the serial run because the serial encoding includes constraint (7h) of \cite{eifler_safe_2022}, which asserts that the counterexample downset must contain the set $\{1,2,3,4\}$; as discussed above, our divide-and-conquer oracle omits this constraint to keep the encoding isomorphism-invariant.

For $n = 8$, the serial solver took just over three days, whereas every dynamic divide-and-conquer run finished within about four and a half hours; the fastest configuration, $\{60, 5\}$, finished in under two hours, a roughly $38\times$ speedup. Notably, the total time of each configuration is also well below the serial solve time, so the partitioning reduces not only wall-clock time but also the overall work. This suggests that our proposed techniques can indeed effectively decompose the problem and parallelize the solving process. The configuration with a large initial timeout and an aggressive increase rate, $\{60, 5\}$, performed best on $n = 8$, whereas the same initial timeout with a milder rate, $\{60, 2\}$, performed worst; we believe an aggressive rate reduces duplicated work in the case where a child is only slightly simpler than its parent. Figure~\ref{fig:milp-8052-histogram} shows a breakdown of the solver runtime on the leaf subproblems of the $\{5, 2\}$ configuration. The distribution is heavily right-skewed: the median leaf solve time is about 14 seconds, but the hardest leaf required more than an hour. This suggests that a more intelligent dynamic re-splitting scheme could better leverage parallel resources and lead to further performance gain.

\begin{table}[t!]
    \caption{Effect of symmetry-breaking constraints on $n = 8$ with tuned SCIP settings: wall time, total time, and number of leaf subproblems for each configuration of Table~\ref{tab:milp-times-tuned}, without and with the ILP analogue of the symmetry constraints $\sigma_p$ imposed on each subproblem.}
    \label{tab:milp-sym}
    \setlength{\tabcolsep}{4pt}
    \begin{center}
        \begin{tabular}{l rr rr rr}
            \toprule
            & \multicolumn{2}{c}{Wall} & \multicolumn{2}{c}{Total} & \multicolumn{2}{c}{Leaves} \\
            \cmidrule(lr){2-3} \cmidrule(lr){4-5} \cmidrule(lr){6-7}
            Config. & {base} & {+sym} & {base} & {+sym} & {base} & {+sym} \\
            \midrule
            \{5, 2\}  & \text{2:33:11} & \text{2:42:52} & \text{27:26:58} & \text{33:03:47} & 197 & 209 \\
            \{5, 5\}  & \text{2:15:59} & \text{2:27:04} & \text{6:45:03}  & \text{8:07:36}  & 76  & 85 \\
            \{60, 2\} & \text{4:11:04} & \text{4:22:47} & \text{12:40:15} & \text{18:43:41} & 53  & 61 \\
            \{60, 5\} & \text{1:53:49} & \text{1:47:38} & \text{5:08:38}  & \text{6:04:15}  & 33  & 35 \\
            \bottomrule
        \end{tabular}
    \end{center}
\end{table}

Table~\ref{tab:milp-sym} reports runs of the same configurations on $n = 8$ in which each subproblem additionally imposes the ILP analogue of the symmetry constraints $\sigma_p$. Adding these constraints does not boost performance: wall times remain similar (three configurations run slightly slower and one slightly faster), while the total time and the number of leaf subproblems increase in every configuration.

Table~\ref{tab:milp-times-default} reports the results of the same experiments run with SCIP's default settings, using switch depth $s = 10$ for $n = 7$ and $s = 16$ for $n = 8$. The contrast with the serial solver is starker in this setting: serial exact SCIP verifies $n = 7$ in 5 minutes and 24 seconds, and on $n = 8$ it timed out after seven days, while every dynamic configuration verified $n = 8$ in under 21 hours. Dynamic partitioning thus renders the problem tractable even without tuning the solver settings.

\subsection{Proving $n=8$}

Beyond solving the case $n = 8$, we also produced and checked proof certificates for it. To do so, we re-ran the $\{60, 5\}$ configuration with the tuned SCIP settings on the same cluster, this time with SCIP's proof-production options turned on: as before, each node was pruned if it was solved within its timeout and branched otherwise. The run partitioned the problem into 38 leaf subproblems, and we successfully extracted a proof certificate in the VIPR format \cite{cheung_verifying_2017} for each of them. Producing certificates roughly doubles the cost of solving: the run took 4:13:42 of wall time and 11:27:21 of total solving and branching time, compared to 1:53:49 and 5:08:38 for the same configuration without proof production. The overhead also slightly alters the shape of the search: the run produced 38 leaf subproblems rather than the 33 of the certificate-free run, as slower oracle calls push a few additional nodes past their timeouts and cause them to be branched. We then ran the official proof checker \cite{cheung_verifying_2017} on those certificates and successfully validated all of them; checking took 1970 seconds of CPU time in total, and the longest proof-checking time for a single certificate was 356 seconds. The total size of the certificates is 14\,GB. The proof certificates and the proof-checking logs are publicly available on Zenodo \cite{looney_vipr-proof}. We have also checked that the leaf nodes form a partition of the original problem, so the validated certificates together constitute a machine-checked proof of Chvátal's Conjecture for $n = 8$. In contrast, prior work solved $n = 8$ with serial exact SCIP without generating a proof certificate \cite{eifler_algorithms_2025} and noted that a certificate of such a sequential run would exceed 1\,TB, which is infeasible for the proof checker. As far as we know, this is the first time this case of the conjecture has been formally verified.

We note that the certificates cover only the solver's reasoning on the individual subproblems. The partitioning steps themselves, whose correctness is guaranteed by Theorems \ref{thm:symbreak-preserves-counterexamples} and \ref{thm:symbreak-preserves-disjoint}, are not accompanied by certificates; conceptually, the partitioner carries out part of the deductive work that a monolithic run would have to certify, which also helps explain why our certificates are far smaller than the estimated 1\,TB. Unlike the numerical reasoning that VIPR certifies, this part of the reasoning is discrete and involves no floating-point arithmetic; still, producing machine-checkable certificates for the partitioning itself would make the proof verifiable end to end, and is an appealing direction for future work. It might also be desirable to construct a single certificate of the original problem from certificates of sub-problems~\cite{nair2022proof}.

\section{Conclusion}
We presented a novel partitioning strategy for extremal combinatorics that breaks on isomorphic families of sets. Our method of branching on an extension of a partial family of sets and banning isomorphic equivalents both showed improvements over the look-ahead-based partitioning for SAT and successfully parallelized ILP solving. A dynamic re-splitting scheme that combines our partitioning strategy with an exact MILP solver exactly solved Chvátal's Conjecture for $n = 8$ in a reasonable time, and produced the first verified proof of this case of the statement.

In the future, it would be interesting to test our methodology on other problems in extremal set theory, such as the Union-closed Sets Conjecture \cite{bruhn_journey_2015}. Our method is already general enough to handle these cases, so extensions would merely need to formally describe an appropriate predicate in a language that admits automatic deduction in order to apply our approach. Using our solver-agnosticism, such explorations could also investigate whether ILP or SAT methods are superior on the chosen problems. A related direction is to investigate why ILP is more effective than SAT in the case of Chvátal's Conjecture. Our implementation could also be pushed further by better utilizing available cores with more intelligent dynamic re-splitting. Considering that solver configurations significantly impacted runtime in our experiments, performing parameter tuning~\cite{wu2023lightweight,wilson2025per,wu2025cubingtuning} may further boost performance. Finally, producing machine-checkable certificates for the partitioning steps themselves would make the resulting proofs verifiable end to end.

\section*{Acknowledgment}
This work was performed in part using high-performance computing equipment at Amherst College obtained under National Science Foundation Grant No. 2117377. Any opinions, findings, and conclusions or recommendations expressed in this publication are those of the authors and do not necessarily reflect the views of the National Science Foundation. Wu's research is supported by a gift from The VMware University Research Fund.

We are grateful to the anonymous IJCAR and FMCAD reviewers, who together substantially improved the quality of this contribution with their thoughtful and constructive feedback.

\bibliographystyle{IEEEtran}
\bibliography{references.bib}

@article{kirchweger_sat-modulo-sym,
author = {Kirchweger, Markus and Szeider, Stefan},
title = {SAT Modulo Symmetries for Graph Generation and Enumeration},
year = {2024},
issue_date = {July 2024},
publisher = {Association for Computing Machinery},
address = {New York, NY, USA},
volume = {25},
number = {3},
issn = {1529-3785},
url = {https://doi.org/10.1145/3670405},
doi = {10.1145/3670405},
journal = {ACM Trans. Comput. Logic},
month = jul,
articleno = {18},
numpages = {30}
}

@article{eifler_algorithms_2025,
	title = {Algorithms and certificates for exact {Mixed} {Integer} {Programming}},
	url = {https://depositonce.tu-berlin.de/handle/11303/25121},
	language = {en},
	urldate = {2026-07-23},
	author = {Eifler, Leon},
	year = {2025},
}

@book{baader_term-rewriting,
	address = {Cambridge},
	title = {Term Rewriting and All That},
	isbn = {978-0-521-77920-3},
	url = {https://www.cambridge.org/core/books/term-rewriting-and-all-that/71768055278D0DEF4FFC74722DE0D707},
	doi = {10.1017/CBO9781139172752},
	urldate = {2026-07-10},
	publisher = {Cambridge University Press},
	author = {Baader, Franz and Nipkow, Tobias},
	year = {1998},
}

@article{bruhn_journey_2015,
	title = {The {Journey} of the {Union}-{Closed} {Sets} {Conjecture}},
	volume = {31},
	issn = {1435-5914},
	url = {https://doi.org/10.1007/s00373-014-1515-0},
	doi = {10.1007/s00373-014-1515-0},
	language = {en},
	number = {6},
	urldate = {2026-04-24},
	journal = {Graphs and Combinatorics},
	author = {Bruhn, Henning and Schaudt, Oliver},
	month = nov,
	year = {2015},
	pages = {2043--2074},
}

@inproceedings{subercaseaux_automated_2024,
	address = {Cham},
	title = {Automated {Mathematical} {Discovery} and {Verification}: {Minimizing} {Pentagons} in the {Plane}},
	isbn = {978-3-031-66997-2},
	shorttitle = {Automated {Mathematical} {Discovery} and {Verification}},
	doi = {10.1007/978-3-031-66997-2_2},
	language = {en},
	booktitle = {Intelligent {Computer} {Mathematics}},
	publisher = {Springer Nature Switzerland},
	author = {Subercaseaux, Bernardo and Mackey, John and Heule, Marijn J. H. and Martins, Ruben},
	editor = {Kohlhase, Andrea and Kovács, Laura},
	year = {2024},
	pages = {21--41},
}

@article{brakensiek_resolution_2022,
	title = {The {Resolution} of {Keller}’s {Conjecture}},
	volume = {66},
	issn = {1573-0670},
	url = {https://doi.org/10.1007/s10817-022-09623-5},
	doi = {10.1007/s10817-022-09623-5},
	language = {en},
	number = {3},
	urldate = {2026-04-24},
	journal = {Journal of Automated Reasoning},
	author = {Brakensiek, Joshua and Heule, Marijn and Mackey, John and Narváez, David},
	month = aug,
	year = {2022},
	pages = {277--300},
}

@inproceedings{heule_schur_2018,
	title = {Schur {Number} {Five}},
	volume = {32},
	copyright = {Copyright (c)},
	issn = {2374-3468},
	url = {https://ojs.aaai.org/index.php/AAAI/article/view/12209},
	doi = {10.1609/aaai.v32i1.12209},
	language = {en},
	number = {1},
	urldate = {2026-04-24},
	booktitle = {Proceedings of the {AAAI} {Conference} on {Artificial} {Intelligence}},
	author = {Heule, Marijn},
	month = apr,
	year = {2018},
}

@inproceedings{heule_solving_2016,
	address = {Cham},
	title = {Solving and {Verifying} the {Boolean} {Pythagorean} {Triples} {Problem} via {Cube}-and-{Conquer}},
	isbn = {978-3-319-40970-2},
	doi = {10.1007/978-3-319-40970-2_15},
	language = {en},
	booktitle = {Theory and {Applications} of {Satisfiability} {Testing} – {SAT} 2016},
	publisher = {Springer International Publishing},
	author = {Heule, Marijn J. H. and Kullmann, Oliver and Marek, Victor W.},
	editor = {Creignou, Nadia and Le Berre, Daniel},
	year = {2016},
	pages = {228--245},
}

@article{eifler_safe_2022,
	title = {A {Safe} {Computational} {Framework} for {Integer} {Programming} {Applied} to {Chvátal}’s {Conjecture}},
	volume = {48},
	issn = {0098-3500, 1557-7295},
	url = {https://dl.acm.org/doi/10.1145/3485630},
	doi = {10.1145/3485630},
	language = {en},
	number = {2},
	urldate = {2025-03-23},
	journal = {ACM Transactions on Mathematical Software},
	author = {Eifler, Leon and Gleixner, Ambros and Pulaj, Jonad},
	month = jun,
	year = {2022},
	pages = {1--12},
}

@incollection{hutchison_cube_2012,
	address = {Berlin, Heidelberg},
	title = {Cube and {Conquer}: {Guiding} {CDCL} {SAT} {Solvers} by {Lookaheads}},
	volume = {7261},
	isbn = {978-3-642-34187-8 978-3-642-34188-5},
	shorttitle = {Cube and {Conquer}},
	url = {http://link.springer.com/10.1007/978-3-642-34188-5_8},
	language = {en},
	urldate = {2025-05-09},
	booktitle = {Hardware and {Software}: {Verification} and {Testing}},
	publisher = {Springer Berlin Heidelberg},
	author = {Heule, Marijn J. H. and Kullmann, Oliver and Wieringa, Siert and Biere, Armin},
	editor = {Hutchison, David and Kanade, Takeo and Kittler, Josef and Kleinberg, Jon M. and Mattern, Friedemann and Mitchell, John C. and Naor, Moni and Nierstrasz, Oscar and Pandu Rangan, C. and Steffen, Bernhard and Sudan, Madhu and Terzopoulos, Demetri and Tygar, Doug and Vardi, Moshe Y. and Weikum, Gerhard and Eder, Kerstin and Lourenço, João and Shehory, Onn},
	year = {2012},
	doi = {10.1007/978-3-642-34188-5_8},
	note = {Series Title: Lecture Notes in Computer Science},
	pages = {50--65},
}

@inproceedings{biere_cadical_kissat_2024,
	series = {Department of {Computer} {Science} {Report} {Series} {B}},
	title = {{CaDiCaL}, {Gimsatul}, {IsaSAT} and {Kissat} {Entering} the {SAT} {Competition} 2024},
	volume = {B-2024-1},
	booktitle = {Proc. of {SAT} {Competition} 2024 – {Solver}, {Benchmark} and {Proof} {Checker} {Descriptions}},
	publisher = {University of Helsinki},
	author = {Biere, Armin and Faller, Tobias and Fazekas, Katalin and Fleury, Mathias and Froleyks, Nils and Pollitt, Florian},
	editor = {Heule, Marijn and Iser, Markus and Järvisalo, Matti and Suda, Martin},
	year = {2024},
	pages = {8--10},
}

@misc{heule_cnc_2026,
	title = {marijnheule/{CnC}},
	url = {https://github.com/marijnheule/CnC},
	urldate = {2026-04-27},
	author = {Heule, Marijn},
	month = jan,
	year = {2026},
	note = {original-date: 2018-05-29T12:17:39Z},
}

@inproceedings{Jabs2025RustsatLibrarySat,
  title       = {{RustSAT}: {A} Library For {SAT} Solving in Rust},
  author      = {Jabs, Christoph},
  booktitle   = {28th International Conference on Theory and Applications of Satisfiability
  Testing ({SAT} 2025)},
  editor      = {Berg, Jeremias and Nordstr{\"o}m, Jakob},
  year        = {2025},
  volume      = {341},
  publisher   = {Schloss Dagstuhl---Leibniz-Zentrum f{\"{u}}r Informatik},
  series      = {Leibniz International Proceedings in Informatics ({LIPIcs})},
  pages       = {15:1--15:13},
  doi         = {10.4230/LIPIcs.SAT.2025.15},
  eprint      = {2505.15221},
}

@misc{friedgut2016chvatalsconjecturecorrelationinequalities,
      title={Chv\'{a}tal's Conjecture and Correlation Inequalities}, 
      author={Ehud Friedgut and Jeff Kahn and Gil Kalai and Nathan Keller},
      year={2016},
      eprint={1608.08954},
      archivePrefix={arXiv},
      primaryClass={math.CO},
      url={https://arxiv.org/abs/1608.08954}, 
}

@InProceedings{giuseppe_IP-counterexamples_2020,
author="Lancia, Giuseppe
and Pippia, Eleonora
and Rinaldi, Franca",
editor="Kononov, Alexander
and Khachay, Michael
and Kalyagin, Valery A
and Pardalos, Panos",
title="Using Integer Programming to Search for Counterexamples: A Case Study",
booktitle="Mathematical Optimization Theory and Operations Research",
year="2020",
publisher="Springer International Publishing",
address="Cham",
pages="69--84",
isbn="978-3-030-49988-4"
}

@InProceedings{kenter_IP-graph-pebbling_2018,
author="Kenter, Franklin
and Skipper, Daphne",
editor="Kim, Donghyun
and Uma, R. N.
and Zelikovsky, Alexander",
title="Integer-Programming Bounds on Pebbling Numbers of Cartesian-Product Graphs",
booktitle="Combinatorial Optimization and Applications",
year="2018",
publisher="Springer International Publishing",
address="Cham",
pages="681--695",
isbn="978-3-030-04651-4"
}

@article{pulaj_cutting-planes_2017,
	title = {Cutting planes for families implying {Frankl}’s conjecture},
	volume = {89},
	issn = {0025-5718, 1088-6842},
	url = {https://www.ams.org/mcom/2020-89-322/S0025-5718-2019-03461-0/},
	doi = {10.1090/mcom/3461},
	language = {English},
	number = {322},
	urldate = {2026-02-14},
	journal = {Mathematics of Computation},
	author = {Pulaj, Jonad},
	month = mar,
	year = {2020},
	pages = {829--857},
}

@incollection{ralphs_parallel-MILP_2018,
author="Ralphs, Ted
and Shinano, Yuji
and Berthold, Timo
and Koch, Thorsten",
editor="Hamadi, Youssef
and Sais, Lakhdar",
title="Parallel Solvers for Mixed Integer Linear Optimization",
bookTitle="Handbook of Parallel Constraint Reasoning",
year="2018",
publisher="Springer International Publishing",
address="Cham",
pages="283--336",
isbn="978-3-319-63516-3",
doi="10.1007/978-3-319-63516-3_8",
url="https://doi.org/10.1007/978-3-319-63516-3_8"
}

@InProceedings{chvatal_conjecture_1974,
author="Chv{\'a}tal, V.",
editor="Berge, Claude
and Ray-Chaudhuri, Dijen",
title="Intersecting families of edges in hypergraphs having the hereditary property",
booktitle="Hypergraph Seminar",
year="1974",
publisher="Springer Berlin Heidelberg",
address="Berlin, Heidelberg",
pages="61--66",
isbn="978-3-540-37803-7"
}

@article{parczyk_computer-assisted-combinatorics_2023, title={Fully Computer-Assisted Proofs in Extremal Combinatorics}, volume={37}, url={https://ojs.aaai.org/index.php/AAAI/article/view/26470}, DOI={10.1609/aaai.v37i10.26470}, abstractNote={We present a fully computer-assisted proof system for solving a particular family of problems in Extremal Combinatorics. Existing techniques using Flag Algebras have proven powerful in the past, but have so far lacked a computational counterpart to derive matching constructive bounds. We demonstrate that common search heuristics are capable of finding constructions far beyond the reach of human intuition. Additionally, the most obvious downside of such heuristics, namely a missing guarantee of global optimality, can often be fully eliminated in this case through lower bounds and stability results coming from the Flag Algebra approach. To illustrate the potential of this approach, we study two related and well-known problems in Extremal Graph Theory that go back to questions of Erdős from the 60s.
Most notably, we present the first major improvement in the upper bound of the Ramsey multiplicity of K_4 in 25 years, precisely determine the first off-diagonal Ramsey multiplicity number, and settle the minimum number of independent sets of size four in graphs with clique number strictly less than five.}, number={10}, journal={Proceedings of the AAAI Conference on Artificial Intelligence}, author={Parczyk, Olaf and Pokutta, Sebastian and Spiegel, Christoph and Szabó, Tibor}, year={2023}, month={Jun.}, pages={12482-12490} }

@article{hales_kepler_2017, title={A Formal Proof of the Kepler Conjecture}, volume={5}, DOI={10.1017/fmp.2017.1}, journal={Forum of Mathematics, Pi}, author={Hales, Thomas and Adams, Mark and Bauer, Gertrud and Dang, Tat Dat and Harrison, John and Hoang, Le Truong and Kaliszyk, Cezary and Magron, Victor and McLaughlin, Sean and Niguyen, Tat Thang and et al.}, year={2017}, pages={e2}}

@misc{czabarka2017chvatalsconjecturedownsetssmall,
      title={Chv\'atal's conjecture for downsets of small rank}, 
      author={Eva Czabarka and Glenn Hurlbert and Vikram Kamat},
      year={2017},
      eprint={1703.00494},
      archivePrefix={arXiv},
      primaryClass={math.CO},
      url={https://arxiv.org/abs/1703.00494}, 
}

@InProceedings{cdcl_sym,
author="Metin, Hakan
and Baarir, Souheib
and Colange, Maximilien
and Kordon, Fabrice",
editor="Beyer, Dirk
and Huisman, Marieke",
title="CDCLSym: Introducing Effective Symmetry Breaking in SAT Solving",
booktitle="Tools and Algorithms for the Construction and Analysis of Systems",
year="2018",
publisher="Springer International Publishing",
address="Cham",
pages="99--114",
}

@misc{satsuma,
      title={satsuma: Structure-based Symmetry Breaking in SAT}, 
      author={Markus Anders and Sofia Brenner and Gaurav Rattan},
      year={2024},
      eprint={2406.13557},
      archivePrefix={arXiv},
      primaryClass={cs.DS},
      url={https://arxiv.org/abs/2406.13557}, 
}

@article{margot_pruning_2002,
	title = {Pruning by isomorphism in branch-and-cut},
	volume = {94},
	issn = {1436-4646},
	url = {https://doi.org/10.1007/s10107-002-0358-2},
	doi = {10.1007/s10107-002-0358-2},
	journal = {Mathematical Programming},
	author = {Margot, François},
	month = dec,
	year = {2002},
	pages = {71--90},
}

@article{nauty,
    title = {Practical graph isomorphism, II},
    volume = {60},
    url = {https://doi.org/10.1016/j.jsc.2013.09.003},
    doi = {10.1016/j.jsc.2013.09.003},
    journal = {Journal of Symbolic Computation},
    author = {McKay, Brendan and Piperno, Adolfo},
    year = {2014},
    pages = {94--112},
}

@misc{hojny2025scipoptimizationsuite10,
      title={The SCIP Optimization Suite 10.0}, 
      author={Christopher Hojny and Mathieu Besançon and Ksenia Bestuzheva and Sander Borst and João Dionísio and Johannes Ehls and Leon Eifler and Mohammed Ghannam and Ambros Gleixner and Adrian Göß and Alexander Hoen and Jacob von Holly-Ponientzietz and Rolf van der Hulst and Dominik Kamp and Thorsten Koch and Kevin Kofler and Jurgen Lentz and Marco Lübbecke and Stephen J. Maher and Paul Matti Meinhold and Gioni Mexi and Til Mohr and Erik Mühmer and Krunal Kishor Patel and Marc E. Pfetsch and Sebastian Pokutta and Chantal Reinartz Groba and Felipe Serrano and Yuji Shinano and Mark Turner and Stefan Vigerske and Matthias Walter and Dieter Weninger and Liding Xu},
      year={2025},
      eprint={2511.18580},
      archivePrefix={arXiv},
      primaryClass={math.OC},
      url={https://arxiv.org/abs/2511.18580}, 
}

@InProceedings{bailleux_totalizer_2003,
author="Bailleux, Olivier
and Boufkhad, Yacine",
editor="Rossi, Francesca",
title="Efficient CNF Encoding of Boolean Cardinality Constraints",
booktitle="Principles and Practice of Constraint Programming -- CP 2003",
year="2003",
publisher="Springer Berlin Heidelberg",
address="Berlin, Heidelberg",
pages="108--122",
isbn="978-3-540-45193-8"
}

@incollection{maher_pyscipopt_2016,
  author = {Stephen Maher and Matthias Miltenberger and Jo{\~{a}}o Pedro Pedroso and Daniel Rehfeldt and Robert Schwarz and Felipe Serrano},
  title = {{PySCIPOpt}: Mathematical Programming in Python with the {SCIP} Optimization Suite},
  booktitle = {Mathematical Software {\textendash} {ICMS} 2016},
  publisher = {Springer International Publishing},
  pages = {301--307},
  year = {2016},
  doi = {10.1007/978-3-319-42432-3_37},
}

@InProceedings{painless,
author="Le Frioux, Ludovic
and Baarir, Souheib
and Sopena, Julien
and Kordon, Fabrice",
editor="Gaspers, Serge
and Walsh, Toby",
title="PaInleSS: A Framework for Parallel SAT Solving",
booktitle="Theory and Applications of Satisfiability Testing -- SAT 2017",
year="2017",
publisher="Springer International Publishing",
address="Cham",
pages="233--250",
isbn="978-3-319-66263-3"
}

@InProceedings{paracooba,
author="Heisinger, Maximilian
and Fleury, Mathias
and Biere, Armin",
editor="Pulina, Luca
and Seidl, Martina",
title="Distributed Cube and Conquer with Paracooba",
booktitle="Theory and Applications of Satisfiability Testing -- SAT 2020",
year="2020",
publisher="Springer International Publishing",
address="Cham",
pages="114--122",
isbn="978-3-030-51825-7"
}

@misc{wu2025cubingtuning,
      title={Cubing for Tuning}, 
      author={Haoze Wu and Clark Barrett and Nina Narodytska},
      year={2025},
      eprint={2504.19039},
      archivePrefix={arXiv},
      primaryClass={cs.LO},
      url={https://arxiv.org/abs/2504.19039}, 
}

@misc{klausesclauses,
      title={From Clauses to Klauses}, 
      author={Joseph Reeves and Martin Heule and Randal Bryant},
      year={2024},
      url={https://www.cs.cmu.edu/~jereeves/research/cav24-paper.pdf}, 
}

@inproceedings{cheung_verifying_2017,
	address = {Cham},
	series = {Lecture {Notes} in {Computer} {Science}},
	title = {Verifying {Integer} {Programming} {Results}},
	volume = {10328},
	doi = {10.1007/978-3-319-59250-3_13},
	booktitle = {Integer {Programming} and {Combinatorial} {Optimization} ({IPCO} 2017)},
	publisher = {Springer},
	author = {Cheung, Kevin K. H. and Gleixner, Ambros and Steffy, Daniel E.},
	year = {2017},
	pages = {148--160},
}

@dataset{looney_vipr-proof,
  author       = {Looney, Jesse and
                  Klingler, Allison and
                  McDonald, Jonah and
                  Wu, Gloria and
                  Pulaj, Jonad and
                  Wu, Haoze},
  title        = {Cube-and-Conquer VIPR Proof of Chvátal's
                   Conjecture for Ground Sets of Size 8
                  },
  month        = jul,
  year         = 2026,
  publisher    = {Zenodo},
  doi          = {10.5281/zenodo.21603128},
  url          = {https://doi.org/10.5281/zenodo.21603128},
}

@inproceedings{wilson2025per,
  title={Per-Instance Subproblem Generation for Strategy Selection in SMT},
  author={Wilson, Amalee and Narodytska, Nina and Barrett, Clark and Wu, Haoze},
  booktitle={Formal Methods in Computer Aided Design (FMCAD)},
  pages={94--103},
  year={2025},
  organization={TU Wien Academic Press}
}

@inproceedings{wu2023lightweight,
  title={Lightweight Online Learning for Sets of Related Problems in Automated Reasoning},
  author={Wu, Haoze and Hahn, Christopher and Lonsing, Florian and Mann, Makai and Ramanujan, Raghuram and Barrett, Clark},
  booktitle={2023 Formal Methods in Computer-Aided Design (FMCAD)},
  pages={1--11},
  year={2023},
  organization={IEEE}
}

@inproceedings{nair2022proof,
  title={Proof-stitch: Proof combination for divide-and-conquer sat solvers},
  author={Nair, Abhishek and Chattopadhyay, Saranyu and Wu, Haoze and Ozdemir, Alex and Barrett, Clark},
  booktitle={2022 Formal Methods in Computer-Aided Design (FMCAD)},
  pages={84--88},
  year={2022},
  organization={IEEE}
}

\newpage
\onecolumn
\appendix
\renewcommand\thetable{\thesection.\arabic{table}}
\renewcommand\thetheorem{\thesection.\arabic{theorem}}
\section{Deferred Proofs}
\label{section:proofs}

\restate{thm:symbreak-preserves-counterexamples}
\begin{proof}
    Induction on $k$. When $k=0$, the claim is trivial. Suppose the claim holds for some $k \in \nat$. Then it suffices to show $\phi_\exists(\Nodes_{k+1}^\#) = \phi_\exists(\Nodes_k^\#)$ because $\phi_\exists(\Nodes_k^\#) = \phi_\exists(\Nodes_0^\#)$. Consider the $k$th rule application (which transformed $\Nodes_k$ into $\Nodes_{k+1}$).
    
    If the \prune{} rule was used with target node $N' \in \Nodes_k$, we have $\phi_\exists(N'^\#) = false$ and
    \begin{align}
        \Nodes_k^\#
            = \bigcup_{N \in \Nodes_k} N^\#
            = N'^\# \cup \bigcup_{N \in \Nodes_k \setminus \{N'\}} N^\#
            = N'^\# \cup \Nodes_{k+1}^\#.
    \end{align}
    Therefore,
    \begin{align}
        \phi_\exists(\Nodes_k^\#)
            = \phi_\exists(N'^\# \cup \Nodes_{k+1}^\#)
            = \phi_\exists(N'^\#) \lor \phi_\exists(\Nodes_{k+1}^\#)
            = false \lor \phi_\exists(\Nodes_{k+1}^\#)
            = \phi_\exists(\Nodes_{k+1}^\#).
    \end{align}

    Otherwise, the \branch{} rule was used with target node $N' = \node{f}{B} \in \Nodes_k$, creating children $N_1 = \node{p}{B}$ and $N_2 \node{f}{B \cup \{p\}}$. Using Lemma \ref{lem:children-union} to turn $N'^\#$ into $N_1^\# \cup N_2^\#$, we have
    \begin{align}
        \Nodes_k^\#
            = N'^\# \cup \bigcup_{N \in \Nodes_k \setminus \{N'\}} N^\#
            = N_1^\# \cup N_2^\# \cup \bigcup_{N \in \Nodes_k \setminus \{N'\}} N^\#
            = \bigcup_{N \in (\Nodes_k \setminus \{N'\}) \cup \{N_1, N_2\}} N^\#
            = \Nodes_{k+1}^\#
    \end{align}
    so $\phi_\exists(\Nodes_{k+1}^\#) = \phi_\exists(\Nodes_k^\#)$. By induction, the claim holds for all $k \in \nat$.
\end{proof}

\restate{thm:symbreak-preserves-disjoint}
\begin{proof}
    Induction on $k$. When $k=0$, the claim is trivial. Suppose the claim holds for some $k \in \nat$. Suppose also that the nodes in $\Nodes_0$ have pairwise disjoint concretizations. Then the nodes in $\Nodes_k$ have pairwise disjoint concretizations by the induction hypothesis. We will show that the same is true for $\Nodes_{k+1}$. Consider the $k$th rule application (which transformed $\Nodes_k$ into $\Nodes_{k+1}$).

    If the \prune{} rule was used with target node $N \in \Nodes_k$, the result is trivial: $\Nodes_{k+1} = \Nodes_k \setminus \{N\}$, so the remaining nodes in $\Nodes_{k+1}$ must still have pairwise disjoint concretizations.

    Otherwise, the \branch{} rule was used with target node $N = \node{f}{B} \in \Nodes_k$, creating children $N_1 = \node{p}{B}$ and $N_2 = \node{f}{B \cup \{p\}}$. By Lemma \ref{lem:children-union}, $N_1^\# \cup N_2^\# = N^\#$. In particular, $N_1^\#, N_2^\# \subseteq N^\#$. Therefore, since $N^\#$ is disjoint with the concretization of anything in $\Nodes_k \setminus \{N\}$, $N_1^\#$ and $N_2^\#$ must have the same property. Consequently, in order for $\Nodes_{k+1} = (\Nodes_k \setminus \{N\}) \cup \{N_1, N_2\}$ to have pairwise disjoint concretizations, we only need that $N_1^\#$ is disjoint with $N_2^\#$. Fortunately, Lemma \ref{lem:children-intersection} guarantees exactly that. By induction, the claim holds for all $k \in \nat$.
\end{proof}

\restate{thm:symbreak-terminates}
\begin{proof}
    Suppose for contradiction that there is an infinite sequence of rule applications. Denote the rule sequence by $(r_k)_{k=1}^\infty$, and denote the companion sequence of states of $\Nodes$ by $(\Nodes_k)_{k=0}^\infty$, where $\Nodes_0$ is the initial state and each $\Nodes_{k+1}$ is the state achieved by applying $r_{k+1}$ in state $\Nodes_k$. Let $m$ be the maximum magnitude among the nodes in $\Nodes_0$. To each state $\Nodes_k$, we assign a tuple $X^k = (x^k_m, \dots, x^k_0) \in \nat^{m+1}$, where $x^k_i$ is the number of nodes of magnitude $i$ in $\Nodes_k$.
    
    Since $(\nat, >)$ is terminating (i.e. well-ordered), $(\nat^{m+1}, >_{lex})$ is terminating, where $>_{lex}$ is the $(m+1)$-fold lexicographic product of $>$. Hence, there exist no infinite descending chains in $(\nat^{m+1}, >_{lex})$. (See \cite[pp.~7--19]{baader_term-rewriting} for these notions and results.) However, we will show that the infinite sequence $(X^k)_{k=0}^\infty$ is strictly decreasing, producing a contradiction.
    
    Claim 1: For all $k \in \nat$, the magnitudes of the nodes in $\Nodes_k$ are bounded above by $m$.
    Proof: Induction on $k$. When $k=0$, the claim follows from our choice of $m$. If the claim holds for some $k \in \nat$, then consider $r_{k+1}$. If the \prune{} rule was used, then $\Nodes_{k+1} \subseteq \Nodes_k$, so the magnitudes of the nodes in $\Nodes_{k+1}$ must still be bounded above by $m$. If the \branch{} rule was used, then $\Nodes_{k+1}$ is the same as $\Nodes_k$ except that a node has been replaced by two nodes with smaller magnitudes (by Lemma \ref{lem:branch-decreases-magnitude}). Hence, the magnitudes of the nodes in $\Nodes_{k+1}$ must still be bounded above by $m$. By induction, the claim holds for all $k \in \nat$.
    
    We will now show that the infinite sequence $(X^k)_{k=0}^\infty$ is strictly decreasing, as required. Let $k \in \nat$ and consider $r_{k+1}$. Call the target node of this rule application $N \in \Nodes_k$, let $i = \norm{N}$, and note that $i \in [0, m]$ by Claim 1 and the definition of magnitude. If the \prune{} rule was used, then $X^{k+1} = (x^k_m, \dots, x^k_i - 1, \dots, x^k_0) <_{lex} X^k$. If the \branch{} rule was used, let $i_1$ and $i_2$ denote the magnitudes of the children created by the branch, and note that $i_1, i_2 < i$ by Lemma \ref{lem:branch-decreases-magnitude} (say WLOG that $i_1 \leq i_2$). Therefore,
    \begin{align}
        X^{k+1} &= (x^k_m, \dots, x^k_i - 1, \dots, x^k_{i_2} + 1, \dots, x^k_{i_1} + 1, \dots, x^k_0) \\
        \text{or} \notag\\
        X^{k+1} &= (x^k_m, \dots, x^k_i - 1, \dots, x^k_{i_2} + 2, \dots, x^k_0)
    \end{align}
    depending on whether $i_1 = i_2$ or not. In either case, $X^{k+1} <_{lex} X^k$ because the additions to $x_{i_1}$ and $x_{i_2}$ occur at strictly later indices than the subtraction from $x_i$. This is true for any $k \in \nat$, so $(X^k)_{k=0}^\infty$ is strictly decreasing.
\end{proof}

\restate{thm:prune-formula-sym}
\begin{proof}
$(\implies)$ Suppose $\phi_\exists(N^\#)$. Then there is a family $g \in N^\# = \compiso(f) \setminus \compiso(B)$ that satisfies $\phi(g)$. Since $g \in \compiso(f)$, we have $g \supseteq f'$ for some $f' \in [f]$ (so there exists a permutation $\pi_f$ that sends $f'$ to $f$). Let $g_1 = \pi_f(g) \in [g]$, so that $g_1 \supseteq \pi_f(f') = f$.

    Our goal is to find some $g_2 \in [g]$ such that $f \subseteq g_2$ and $g_2 \in \compiso(p) \implies p \subseteq g_2$. There are two cases:

    \textbf{Case 1:} $g_1 \notin \compiso(p)$. Then we simply let $g_2 = g_1 \in [g]$. We have $g_2 = g_1 \supseteq f$ and the additional implication, so $g_2$ is as desired.

    \textbf{Case 2:} $g_1 \in \compiso(p)$. Then $p' \subseteq g_1$ for some $p' \in [p]$. There is a permutation $\pi_p$ that sends $p'$ to $p$. Let $g_2 = \pi_p(g_1) \in [g]$, so that $g_2 \supseteq \pi_p(p') = p$. Since $p \in \extiso(f)$, we have $p \supseteq f$, so $g_2 \supseteq f$, and $g_2$ is as desired.

    In either case, we proceed with the $g_2$ that we found. Since $\phi$ is isomorphism-invariant, we have $\phi(g_2) = \phi(g) = true$. Lemma \ref{lem:banned-formula} gives $g_2 \notin \compiso(B)$ because $g \notin \compiso(B)$. We also have $f \subseteq g_2$ and $g_2 \in \compiso(p) \implies p \subseteq g_2$ by construction. Therefore, $\phi_N^p(g) = true$, so $\phi_N^p$ is satisfiable.
    
    $(\impliedby)$ Suppose there is a family $g$ that satisfies $\phi_N^p$. In particular, $g$ must satisfy $\phi_N$. Since $\phi_N$ is satisfiable, Theorem \ref{thm:prune-formula} entails that $\phi_\exists(N^\#)$.
\end{proof}

The following theorem is not stated in the main text, but Lemma \ref{lem:no-branch-if-f-banned} actually gives us the surprising result that branching never creates a family with a banned current family. That means we will never create empty nodes---instead, terminal nodes must necessarily have all their extensions banned.

\setcounter{theorem}{0}
\begin{theorem}
    \label{thm:branch-never-bans-f}
    Let $N_1 = \node{p}{B}$ and $N_2 = \node{f}{B \cup \{p\}}$ be the child nodes spawned by applying the \branch{} rule to a node $N = \node{f}{B}$ in a proof system $\prfsys(\phi, \extiso, \compiso)$. Then $p \notin \compiso(B)$ and $f \notin \compiso(B \cup \{p\})$.
\end{theorem}
\begin{proof}
    Suppose for contradiction that $p \in \compiso(B)$. Then $p \in \compiso(b)$ for some $b \in B$. By Lemma \ref{lem:compiso-f-in-compiso-g}, that means $\compiso(p) \subseteq \compiso(b) \subseteq \compiso(B)$. But since we applied the \branch{} rule, we must have had $\compiso(p) \not\subseteq \compiso(B)$, a contradiction. Hence, $p \notin \compiso(B)$.

    Suppose for contradiction that $f \in \compiso(p)$. By Lemma \ref{lem:compiso-f-in-compiso-g}, that means $\compiso(f) \subseteq \compiso(p)$. But $p \in P \subseteq \extiso(f)$, so by Property \ref{eq:ext-comp-property}, we have $\compiso(p) \subset \compiso(f)$, a contradiction. Hence, $f \notin \compiso(p)$. Moreover, $f \notin \compiso(B)$ by Lemma \ref{lem:no-branch-if-f-banned}, so we have $f \notin \compiso(B) \cup \compiso(p) = \compiso(B \cup \{p\})$. 
\end{proof}

\end{document}